\documentclass[10pt,twocolumn]{IEEEtran}
\usepackage{cite}
\usepackage{amsmath,amssymb,amsfonts}
\usepackage{algorithmic}
\usepackage{graphicx}
\usepackage{subcaption}
\usepackage{textcomp}
\usepackage{xcolor}
\usepackage{multirow}
\newtheorem{problem}{Problem}
\newtheorem{lemma}{Lemma}

\newtheorem{remark}{Remark}
\newtheorem{corollary}{Corollary}
\newtheorem{theorem}{Theorem}
\newenvironment{proof}{\paragraph{Proof:}}{\hfill$\square$}
\newenvironment{proofie}{\paragraph{Proof of Theorem 1:}}{\hfill$\square$}

\def\BibTeX{{\rm B\kern-.05em{\sc i\kern-.025em b}\kern-.08em
    T\kern-.1667em\lower.7ex\hbox{E}\kern-.125emX}}
\begin{document}
\title{Minimizing the Risk of Spreading Processes via Surveillance Schedules and Sparse Control}
\author{Vera L. J. Somers and Ian R. Manchester
\thanks{The authors are with the Australian Centre for Field Robotics (ACFR), School of Aerospace, Mechanical and Mechatronic Engineering,
        University of Sydney, NSW 2006, Australia
        {\tt\small \{v.somers, i.manchester\}@acfr.usyd.edu.au}}}

\maketitle

\begin{abstract}
In this paper, we propose an optimization framework that combines surveillance schedules and sparse control to bound the risk of spreading processes such as epidemics and wildfires. Here, risk is considered the risk of an undetected outbreak, i.e. the product of the probability of an outbreak and the impact of that outbreak, and we can bound or minimize the risk by resource allocation and persistent monitoring schedules. The presented framework utilizes the properties of positive systems and convex optimization to provide scalable algorithms for both surveillance and intervention purposes. We demonstrate with different spreading process examples how the method can incorporate different parameters and scenarios such as a vaccination strategy for epidemics and the effect of vegetation, wind and outbreak rate on a wildfire in persistent monitoring scenarios. 
\end{abstract}


\section{Introduction}
\label{sec:introduction}
\IEEEPARstart{E}{pidemics} and wildfires claim each year a significant number of lives, have major impacts on the economy and can cause signification disruptions to life as we know it as the outbreak of COVID-19 and the Australian and Californian wildfire seasons of 2019/2020 have demonstrated  \cite{jolly2015climate,Burrows2019,Hughes2016,peeri2020sars,gasmi2020individual}. What they both have in common is that they can be thought of and modeled as a spreading process in which an initial localized outbreak spreads rapidly to neighboring nodes in a network \cite{Karafyllidis1997a,Bloem2008,Nowzari2016}.

The real-world risks associated with these events and other spreading processes, such as computer viruses, have sparked significant research into methods for modeling, prediction and control, see the recent surveys \cite{Nowzari2016,pare2020modeling,zino2021analysis}. The aftermath of two decades of multiple deadly coronavirus epidemics alone (SARS, MERS, COVID-19) stresses the importance and (current shortcoming) of being able to not only quickly map and monitor the spread of infection, but more importantly get an accurate risk estimate to plan appropriate intervention strategies such as lockdowns, testing and vaccination \cite{peeri2020sars}. For wildfires specifically, there is an interest in remote sensing and in particular to connect existing fire propagation models and simulators (e.g. \cite{Karafyllidis1997a,  Johnston2006}) to path planning for spatial monitoring (e.g. \cite{smith2012persistent, penicka2017dubins}) to accurately estimate and assess the fire risk \cite{Burrows2019,Yebra2018}. 

The large scale these spreading processes typically evolve on imposes another challenge and emphasizes the importance of scalability of computational methods. Furthermore, sparsity of resource allocation solutions is often needed, because it can be difficult to distribute resources widely, such as vaccines for epidemics or waterbombing allocations for wildfires. 

In this paper we provide a flexible optimization framework to bound and minimize the risk of spreading processes via surveillance schedules and sparse resource allocation. The presented method distinguishes itself from previous research by: 1) providing the joint optimization of both surveillance and intervention via sensor scheduling and sparse control, 2) proposing a risk model with node dependent costs, 3) sparsity inducing resource models for spreading, recovery and outbreak rate and 4) scalability due to use of convex programming. For wildfires specific we bridge the gap between spreading models and robotic path planning.

\subsection{Relevant Literature}
\subsubsection{Spreading Processes}
Spreading processes are commonly modeled as Markov processes. The most well-known models, coming from epidemics literature, are the Susceptible-Infected-Susceptible (SIS) model and the Susceptible-Infected-Removed (SIR) model and their variants \cite{kermark1927contributions,bailey1975mathematical}. These stochastic models can be approximated as ordinary differential equation (ODE) models, which can in turn be approximated by a linear model \cite{Ahn2013,VanMieghem:2009,Nowzari2016} which is proven by \cite{VanMieghem:2009} to be an upperbound and is therefore usually the object of study.  

Similar models of firefront propagation are stochastic cellular automata (CA) models, which are based on fire transition probabilities from one cell to another and are most commonly used due to their ease of implementation, low computational cost and suitability for heterogeneous conditions \cite{Sullivan2009}. When based on accurate information about vegetation, terrain, and weather conditions, CA models have been shown to give quantitatively accurate predictions when compared to real wildfire spread \cite{Alexandridis2008a,Kelso2015}. 

Given a spreading process and its known model, there are two main problems that can be considered. First of all, surveillance scheduling, i.e. how to estimate the state of the process. Especially since spreading processes evolve on large networks, a sampling schedule is needed to determine the state of the outbreak. The second problem is intervention, i.e. how to modify the spreading process and reduce the impact of an outbreak. 

Both of these problems have been approached by ranking nodes, for which many different methods exist. The most common ones are based on graph topology, e.g. eigenvector centrality\cite{Youssef2011}. However, in \cite{ Liu2016} is shown that considering topological features alone is in most cases not enough to indicate influential nodes. In wildfire models, it is particularly important to consider spatially-varying spreading rates, due to the effects of vegetation, terrain, and weather. 

A method that goes beyond simple topological features is minimizing the spectral radius, i.e. minimizing the dominant eigenvalue of the linear dynamics, which results in a reduction in overall spreading rate across the network \cite{VanMieghem:2009,van2011decreasing,Preciado2014,Nowzari2017,Youssef2011}. However, in many cases it is important to take into account node-dependent costs for targeted risk minimization. For example, higher cost may be associated with nodes representing populated areas when controlling a wildfire, or more vulnerable members of the community in an epidemic. This is not taken into account when minimizing the spectral radius.

\subsubsection{Surveillance}
Surveillance scheduling, or sensor scheduling, consists of finding an optimal schedule for taking measurements at individual nodes to estimate or influence the state of the total system, given a limited budget on the number of nodes that can be sensed at the same time. In literature sensor scheduling is often considered from a control or theoretical perspective. E.g. sensor and actuator placement to optimize controllabilty and observability for dynamical networks is studied in \cite{Summers2016}. For linear dynamical systems tractable algorithms are developed in \cite{Vitus2012} and a simple example in application for active mapping is presented. 

Sensor scheduling applied to spreading processes on the other hand is more limited. For epidemics, sensor scheduling could be interpreted as a testing schedule or the selection of individuals to be tested for which a interest has arrived since COVID-19 \cite{pezzutto2021smart}. For wildfires, sensor scheduling or more specifically persistent monitoring with use of UAVs is a more active area of research. However most papers focus more on fire mapping and boundary estimation or path planning instead of risk minimization of an undetected outbreak\cite{skeele2016aerial,rabinovich2018toward}.  

\subsubsection{Intervention}
For the intervention problem, the problems of minimizing the spreading rate by removing either a fixed number of links or nodes in the network are both NP-hard \cite{van2011decreasing}. This fact has motivated the study of heuristics methods based on node-rankings of various forms \cite{Hadjichrysanthou2015,Dhingra2018,Lindmark,matsypura2018wildfire}. However, in general these approaches will not be optimal in any sense and furthermore the assumption of complete link or node removal is often unrealistic.

Therefore a more suitable approach and realistic assumption is to reduce the impact of an outbreak by tuning and controlling the spreading rate and recovery rate. This can be achieved by applying resources to the nodes and links. Solutions to the optimal control problem of regulating epidemic spread over networks are presented in \cite{Giamberardino2017,Bloem2008,Khanafer,Mai2018,Liu2019b,Torres2017}. More recently limitations on feedback control have been identified in \cite{ye2020distributed,ye2021applications}. It is proven that a large class of distributed controllers cannot guarantee convergence to the disease-free equilibrium utilizing the recovery rate as control input. 

Linear spreading processes are positive systems \cite{berman1994nonnegative}, which enables control based on linear programming (LP) \cite{briat2013robust,rantzer2015scalable} and geometric programming (GP) \cite{Boyd2007}, which allows global optima to be found with efficient numerical methods and significantly improved scalability compared to general linear systems.  Optimal resource allocation via geometric programming has been studied in \cite{Preciado2014,Zhang2018,Nowzari2017, Ogura2019}. However, most of these do not result in sparse resource allocation. In general, the problem of designing  sparse feedback gains for linear systems is non-convex and computationally challenging \cite{Lin2013}. In addition, most of these papers consider minimizing the spectral radius instead of node dependent costs. 

Our proposed method takes into account targeted risk minimization and sensor scheduling and is similar in a sense to GP, that it is convex under logarithmic transformation, in particular it is an exponential cone program. Furthermore, under this transformation our resource model corresponds to $\ell_{1}$ type constraints which are known to encourage sparsity \cite{tibshirani1996regression, candes2006robust, donoho2006compressed} and sparsity can be further increased using reweighted iterations \cite{Candes2008}. Although we focus on spreading processes, the proposed framework is applicable to any linear positive system for which sparse control applications are needed.

\subsection{Paper Structure}
The remainder of the paper is organized as follows. In Section II we discuss the problem and model formulation in detail and introduce respectively the surveillance model in Section III and the resource model in Section IV. The convex optimization framework is presented in Section V. The resulting risk maps, persistent monitoring schedules and resource allocation for both epidemics and wildfire examples are demonstrated in the results Section V. We convey that the proposed method is easy applicable and provides quick and convincing results. 

This paper builds on work presented in \cite{icrapaper2019,LCSS2021}. In \cite{icrapaper2019} the concept of a priority map, a map that defines critical areas or nodes on a graph, was explained in context of a preliminary surveillance problem for wildfires. \cite{LCSS2021} was limited to the intervention scenario by itself with only flexibility in spreading and recovery rate. This paper provides a more flexible framework based on a risk model that besides priority also takes into account outbreak probability, a more extensive resource model including outbreak and revisit rate and more comprehensive results including both surveillance and intervention scenarios.  

\section{MODEL FORMULATION}
The problem that is considered in this paper is to design a surveillance schedule and determine intervention strategies to bound the risk of an undetected outbreak of a spreading process. Here, risk is the product of the probability of an outbreak starting in a particular area and the impact of the process spreading from that location. 

To model this problem, let us consider a network graph $\mathcal{G}$ with $n$ nodes and edge set $\mathcal{E}$, where at each node $i \in \{1, 2, ..., n\}$ an outbreak can occur. We model the probability of an outbreak occurring at a node $i$ within a time interval as a Poisson process \cite[Chapter~2]{ross1996stochastic}. A Poisson process is a counting process with rate $\lambda >0$, in which $N(t)$ represents the total number of events up to time $t$, such that
\begin{enumerate}
\item $N(0)=0$,
\item The process has independent increments, i.e. $P\{N(t+s) - N(s)\}$ for any $s \geq 0$, $t>0$  is independent of $N(s)$,
\item The number of events $N(t)$ happening in any interval of length $t$ is Poisson distributed with mean $\lambda t$. That is, for all $s$, $t\geq 0$
\end{enumerate} 
\begin{equation}
P\left\{N(t+s) - N(s)=m \right\}=e^{-\lambda t} \frac{(\lambda t)^m}{m!} \quad m=0,1,...
\end{equation}

In our model outbreaks occurring at each node $i$ are modeled as independent Poisson processes with rates $\lambda_i$. An outbreak then spreads according to a cellular automata model. Cellular automata simulations such as \cite{Karafyllidis1997a,Johnston2006,Alexandridis2008a} use an underlying stochastic model with Markov transition probabilities and each node $i \in \{1, 2, ..., n\}$ has states $X_i(t)$ (infected) and $Z_i(t)$ (removed) associated with it. We consider a basic susceptible-infected-removed (SIR) model \cite{kermark1927contributions} in which a node can be in three states: infected, i.e., $X_i(t)=1$, $Z_i(t)=0$, susceptible to infection, i.e., $X_i(t)=0$, $Z_i(t)=0$, or removed, i.e., $X_i(t)=0$, $Z_i(t)=1$. 

The transition probabilities of the infected and removed states $X(t)$ and $Z(t)$ associated with the continuous time Markov chain model \cite[Chapter~5]{ross1996stochastic} are defined for a small time period $\Delta_t$ as
\begin{equation}
\label{eq:Stoc}
\begin{split}
Pr\left\{ \right. & \left(X_{i}(t+\Delta_t),Z_{i}(t+\Delta_t)\right) =(1,0)| \\
& \left.\left(X_{i}(t),Z_{i}(t)\right)=(0,0) \right\} = \sum\limits_{j \in  \mathcal{N}_{i}} \beta_{ij}X_{j}(t)\Delta_t + o(\Delta_t)\\
Pr\left\{ \right. & \left(X_{i}(t+\Delta_t),Z_{i}(t+\Delta_t)\right) =(0,1)| \\
& \left.\left(X_{i}(t),Z_{i}(t)\right)=(1,0) \right\} = \delta_{i}\Delta_t + o(\Delta_t)
\end{split}
\end{equation}
where $\beta_{ij}$ is the spreading rate, $\delta_i$ the rate of removal and $\mathcal{N}_{i}$ the neighborhood of node $i$. 
%
We recall
\begin{lemma}
\label{lem:nonlin}
Let us define $E[X_{i}(t)] = P(X_{i}(t))=1 $ as the probability of a node $i$ being infected at time $t$ and $\dot{z}$ as the rate of removal. If $\chi_{i}(0) = E[X_{i}(0)]$ then the upperbound of the stochastic probabilities, i.e.  $\chi_{i}(t) \geq E[X_{i}(t)]$ for all $t \geq 0$, can be established via the following coupled differential equations 
\begin{align}
\dot{\chi}_{i}(t) &=(1-\chi_{i}(t) - z_{i}(t))\sum^{n}_{j=1}\beta_{ij}\chi_{j}(t)-\delta_{i}\chi_{i}(t) \label{eq:nonlin1}\\
\dot{z}_{i}(t) &=\delta_{i}\chi_{i}(t) \label{eq:R}.
\end{align}
\end{lemma}
\begin{proof}
Similar derivations for the homogeneous differential equations for the SIR model are given in \cite{kermark1927contributions} and for heterogeneous equations in \cite{Youssef2011,Nowzari2016,pare2020modeling}. It is proven in \cite{VanMieghem:2009,Youssef2011}, by showing that $E[X_{i}X_{j}] \geq$ 0 for all i $\neq$ j, that $\chi_{i}(t) \geq E[X_{i}(t)]$ for all $t \geq 0$.
\end{proof}
 Here, the main assumption taken is that all the random variables have zero covariance.

A linear model can be obtained by linearizing the deterministic model, i.e. \eqref{eq:nonlin1} and \eqref{eq:R}, around the infection-free equilibrium point ($\mathbf{x}=0$, $\mathbf{z}=0$). Hence, we obtain
\begin{align}
\dot{x}_{i}(t) &=\sum^{n}_{j=1}\beta_{ij}x_{j}(t)-\delta_{i}x_{i}(t) \label{eqL1} \\
\dot{z}_{i}(t) &=\delta_{i}x_{i}(t) \label{eqL2}.
\end{align}
We can combine \eqref{eqL1} for all $i$ as
\begin{equation}
\label{eq1}
   \dot{x}(t)=Ax(t) 
\end{equation}
with initial condition $x(0)$ where $x(t)=[x_{1}(t),...,x_{n}(t)]^{T}$ with $t\geq0$, is the state of the system and the state matrix $A$ is defined by the linearized spreading dynamics 
 \begin{equation}
  \label{eq:epi}
 a_{ij} = 
\begin{cases}
-\delta_{i}  \le 0 &\quad \text{if}\quad  i=j, \\
\beta_{ij} \ge 0&\quad \text{if}\quad  i\neq j, (i,j) \in \mathcal{E}, \\
0  &\quad \text{otherwise}
 \end {cases}
 \end{equation}
where $\mathcal{E}$ is the set of edges of the graph.
\begin{lemma}
\label{lem:Alin}
Let $\chi(t)$ be the solution of \eqref{eq:nonlin1} with initial condition $\chi(0)$ and let $x(t)$ be the solution of \eqref{eq1}. The linear function \eqref{eq1} upperbounds the nonlinear one \eqref{eq:nonlin1}, i.e. $x(t) \geq \chi(t)$ for $t \geq 0$ when $x(0) \geq \chi(0)$. 
\end{lemma}
\begin{proof}
See \cite{Preciado2014}.
\end{proof}

Given an outbreak has occurred, we associate the following cost function with the Markov process 
\begin{equation}
\label{eq:cost1}
J_M(X_i(0))= \int_{0}^\infty e^{-rt}CE[X_{i}(t)] dt 
\end{equation}
where $C= [c_1, ..., c_n]$ is a row vector defining the cost associated with each node $i$, with each $c_i\ge 0$. The discount rate $r \ge 0$ can be tuned to emphasize near-term cost over long-term cost and can be understood as $r=\frac{1}{t_d}$, i.e. $J_M= \int_{0}^\infty e^{-\frac{t}{t_d}}CE[X_{i}(t)] dt$, where $t_d \ge 0$ is the discount time constant. A shorter time constant $t_d$, i.e. a larger discount rate $r$, indicates short-term cost is prioritized. Since $X_i$ is bounded, $r\geq 0$ is finite.

The cost function for the linear model can be defined by
\begin{equation}
\label{eq:cost2}
J_L\left(x(0)\right)=  \int_{0}^\infty e^{-rt} Cx(t)dt.
\end{equation}
where we assume that $r$ is large enough such that $A-rI$ is Hurwitz-stable, i.e. all eigenvalues have negative real parts, and hence, $J_L$ is finite.
\begin{lemma}
\label{lem:cfup}
If $x(0)=\chi(0) = E[X_{i}(0)]$ then the cost of the linear model \eqref{eq:cost2} upperbounds the cost of the nonlinear model, which upperbounds expected cost of the stochastic model \eqref{eq:cost1} for all $t \geq 0$, 
\begin{align}
\label{eq:Jup}
\int_{0}^\infty e^{-rt} Cx(t)dt &\geq \int_{0}^\infty e^{-rt} C\chi(t)dt \nonumber \\
&\geq \int_{0}^\infty e^{-rt}CE[X_{i}(t)] dt.
\end{align}
\end{lemma}
\begin{proof}
From Lemma \ref{lem:nonlin} and \ref{lem:Alin}, we have $x(t) \geq \chi(t) \geq E[X_i(t)] $ and because $C\geq0$, \eqref{eq:Jup} follows. 
\end{proof}
 \begin{lemma}
\label{lem:p}
If $A-rI$ is Hurwitz-stable, the linear cost-to-go has the following structure, 
\begin{equation}
\label{eq:CC}
   \int_{0}^\infty e^{-rt} C x(t)dt = \sum_{i=1}^n p_i x_i(0)
\end{equation}
where $p$ is given by
\begin{equation}
\label{eq:p}
 p^{T}=C(rI-A)^{-1}
\end{equation} 
or by the equivalent linear program (LP)
\begin{equation}
\label{eq:LP1}
\begin{split}
    \text{minimize} & \quad \quad  |p|_{1} \\
    \text{such that} & \quad \quad p \geq 0, \quad \quad p^{T}A - r p^{T} \leq - C.
\end{split}
\end{equation}
$A$ is the state matrix as defined in \eqref{eq:epi}, $C$ is the cost vector, $r$ the discount rate and \eqref{eq:p} is an upper bound for the stochastic cost-to-go. 
\end{lemma}
\begin{proof}
See \cite{icrapaper2019} for proof and discussion. 
\end{proof}

\begin{corollary}
The constraints on the cost function of the linear model \eqref{eq:cost2} and the linear cost-to-go \eqref{eq:CC} are feasible if and only if $A-rI$ is Hurwitz-stable, i.e. 
$\lambda_{max}(A) \leq r$. Therefore minimizing the discount rate $r$ is an upperbound for minimizing the spectral radius 
\begin{equation}
\label{eq:spec}
   J_{\lambda}=\lambda_{max}(A).
\end{equation}
which is common in literature (e.g. \cite{VanMieghem:2009,van2011decreasing,Preciado2014,Nowzari2017,Youssef2011}).
\end{corollary}

\begin{remark}
To obtain the cost-to-go $p$ different cost functions are possible. For example we can use the total number of removed nodes. We can find the total number of removed nodes by taking the integral of \eqref{eq:R} over the interval $[0,T]$. Assuming there are no removed nodes at the start, i.e., $z(0)=0$, this results in
\begin{equation}
\label{eq:RiskR}
z(T) - z(0)= \int_{0}^{T} \dot{z} dt =\int_{0}^{T} \delta_{i}x_{i}(t) dt.
\end{equation}
Including the option for the discount rate $r$, i.e. put a higher cost on nodes that are removed earlier, and using the total number of removed nodes as the cost, we obtain
\begin{equation}
\label{eq:R2}
   \int_{0}^\infty e^{-rt} \delta x(t)dt = \sum_{i=1}^n p_i x_i(0).
\end{equation}
\end{remark}

\section{SURVEILLANCE MODEL}
The risk $R_i$ at each node $i$ is defined as the product of the probability of at least one outbreak occurring within a time interval and the resulting cost of that outbreak \eqref{eq:cost1}. We now want to bound the risk and minimize the cost of an undetected outbreak by use of a revisit schedule. That is, we want to find the largest revisit interval $\tau_i$ for each node $i$, such that the risk of an undetected outbreak remains bounded by some $R_{\text{max}}>0$. We take the assumption that at each revisit we can sense the state of the system, whereas in between visits we receive no information. That is, at time of visiting the risk of an undetected outbreak is reduced back to a small number $\epsilon_R$ for that node. 

The main theoretical result is
\begin{theorem}
\label{th:SP}
Suppose that $A-rI$ is Hurwitz-stable, the risk of an undetected outbreak remains bounded by $R_{\text{max}}$ for any revisit interval $\tau_i= t_{i}(k+1) - t_{i}(k)$ satisfying 
\begin{equation}
\label{eq:bound}
    \tau_i=t_{i}(k+1) - t_{i}(k) \leq \frac{R_{\text{max}}}{p_{i}\lambda_{i}+\epsilon_R}
\end{equation}
where $t_{i}(k)$ is the kth time node $i$ is visited, $\lambda_i$ the Poisson rate of an outbreak occurring at node $i$ and $p$ the node impact. 
\end{theorem}

In order to prove this theorem, we first proof the following lemma. 
\begin{lemma} 
\label{lem:bda}
For any interval $t$ the probability of at least one event occurring within this time interval has a linear upperbound $\lambda_i t$. 
\end{lemma} 
\begin{proof}
 The probability of at least one event occurring within time interval $t$ is a consequence of the Poisson definition and given by
\begin{align}
 P\left\{ N(t+s) - N(s) >0 \right\}&=1 - P\{N(t+s)-N(s)=0\} \nonumber\\
& = 1-e^{-\lambda t} \leq \lambda t.  
\end{align} 
\end{proof}
 For a small interval, this is a tight upper bound. 

\begin{proofie}
Using Lemma \ref{lem:bda} we can upperbound the probability of an outbreak occurring at a node $i$ by $\lambda_{i}t$, i.e the estimated outbreak probability over the interval $t=t_i(k+1)-t_{i}(k)$, where $t_{i}(k)$ is the kth time node $i$ is visited, is
\begin{equation}
\hat{x}_{i}(t)=\lambda_{i}(t_i(k+1)-t_{i}(k)).
\end{equation}
The node impact or cost-to-go $p_{i}$ is given in Lemma \ref{lem:p}. Therefore the risk at each node is bounded by
\begin{equation}
\label{eq:Rlamb}
R_{i}(t) \leq p_{i}\lambda_{i}\left(t_i(k+1)-t_{i}(k)\right) \leq R_{\text{max}} .
\end{equation}
\end{proofie}
\begin{remark}
For epidemics, revisiting can be interpreted as testing. Hence, each time that a node gets tested the risk of an undetected outbreak is reduced. 
\end{remark}

\section{RESOURCE ALLOCATION MODEL}
\label{sec:RA}
In the previous section, we discussed how to bound the risk of an undetected outbreak by use of a revisit schedule. In this section, we consider allocating resources to reduce the risk and associated surveillance workload needed. Resources can be allocated to modify the spreading and recovery rates $\beta_{ij}$ and $\delta_{i}$, the outbreak Poisson rate $\lambda_i$ and the revisit interval $\tau_i$. We assume bounded range of possible values $0 < \underline{\beta}_{ij} \leq \beta_{ij} \leq \overline{\beta}_{ij}$, $0 < \underline{\delta}_{i} \leq \delta_{i} \leq \overline{\delta}_{i}<\overline{\Delta}$, $0 < \underline{\lambda}_{i} \leq \lambda_{i} \leq \overline{\lambda}_{i}$ and $0 < \underline{\tau}_{i} \leq \tau_{i} \leq \overline{\tau}_{i}$. The associated resource model is similar to what we proposed in \cite{LCSS2021},
\begin{align}
&f_{ij}\left(\beta_{ij}\right)=w_{ij}\text{log}\left(\frac{\overline{\beta}_{ij}}{\beta_{ij}}\right),\  g_{i}\left(\overline{\Delta}-\delta_{i}\right)=w_{ii}\text{log}\left(\frac{\overline{\Delta}-\underline{\delta}_{i}}{\overline{\Delta}-\delta_{i}}\right) \nonumber\\
&h_{i}\left(\lambda_{i}\right)=\omega_{\lambda_i}\text{log}\left(\frac{\overline{\lambda}_{i}}{\lambda_{i}}\right), \  \psi_{i}\left(\tau_{i}\right)=\omega_{\tau_i}\text{log}\left(\frac{\overline{\tau}_{i}}{\tau_{i}}\right)  \label{eq:RM}
\end{align}
where $w_{ij}$, $w_{ii}$, $\omega_{\lambda_i}$ and $\omega_{\tau_i}$ are weightings expressing the cost of respectively reducing $\beta_{ij}$, increasing $\delta_{i}$ and reducing $\lambda_i$ and $\tau_i$. 

This type of logarithmic resource allocation was discussed in \cite{LCSS2021} for $\beta$ and $\delta$. These resource models can be understood as that a certain proportional increase/decrease always has the same cost. Note that it is impossible for these to be reduced to $0$ since that would take infinite resources. For the revisit interval $\tau$ this can be thought of as that it is impossible to constantly monitor each node $i$, hence there will always be a nonzero interval rate between measurements. Furthermore these resource models encourage sparsity as detailed further in Section \ref{subsec:L0}.

\section{CONVEX OPTIMIZATION}
\label{sec:CO}
In this section we set up the optimization framework and demonstrate how we can reformulate our constraints and objectives to obtain a convex optimization problem. In total there are six different quantities we may want to minimize or bound: risk, resources on spreading rate, recovery rate and outbreak rate, discount rate and surveillance workload, i.e. frequency of revisits.

We can study multiple problems, where the goal is to either minimize the risk or keep the risk bounded, while minimizing resources.  These different objectives or costs, are summarized in Table \ref{tab:cost}. 

\begin{table}
\caption{Spreading process quantities and their convex representation.}
\label{tab:cost}
\centering
\begin{tabular}{ c|c | c} 
 Objective & Model Variable  & Convex Representation  \\ 
\hline
\multirow{2}{*} {Risk} & \multirow{2}{*} {$R_i=p_i \lambda_i \tau_i$} & $y_i + \text{log}(\overline{\lambda}_i) - \frac{z_i}{\omega_{\lambda_i}} $ \\
& & $+ \text{log}(\overline{\tau}_i) - \frac{\sigma_i}{\omega_{\tau_i}}$ \\
Spreading rate &  $ \sum_{ij} f_{ij}\left( \beta_{ij}\right )$ & $\sum_{ij} u_{ij}$ \\
Recovery rate &  $\sum_{i} g_{i}\left(\delta_{i}\right)$ & $\sum_{i} v_{i}$\\
Outbreak rate & $\sum_{i} h_{i}\left(\lambda_{i}\right)$ & $\sum_{i} z_i$\\
Surveillance workload &  $\sum_{i}\psi_i \left(\tau_i\right)$ & $\sum_{i} \sigma_i$ \\
Spectral radius bound & $r$ & $\text{log}(e^\rho-\overline{\Delta})$ \\
\end{tabular}
\end{table}

The constraints the optimization will have to adhere to are summarized in Table \ref{tab:const} and consist of the bounds on the different variables and the dynamic coupling constraint, derived from \eqref{eq:p}. All objectives and constraints can now be reformulated as convex constraints under the following logarithmic transformations; 
\begin{align}
&y_i=\text{log}(p_i), \quad u_{ij}=f_{ij}(\beta_{ij}), \quad v_i=g_i(\overline{\Delta}-\delta_i), \nonumber \\
&z_i=h_i(\lambda_i), \quad \sigma_i= \psi_i(\tau_i), \quad \rho = \text{log}(\overline{\Delta}+r), \label{eq:Convex}
\end{align}
as shown in the right column of Tables \ref{tab:cost} and \ref{tab:const}.

\begin{lemma}
The dynamic coupling constraint $p^{T}A - r p^{T} \leq - C$ is equivalent to the following convex constraint under the transformations \eqref{eq:Convex}
\begin{align}
& q(y,u,v,\rho) = \nonumber \\
&  \text{log} \Biggl (\sum_{i: (i,j)\in \mathcal E} \text{exp} \left(y_{i} + \text{log}\left(\overline{\beta}_{ij}\right) - \frac{u_{ij}}{w_{ij}} -  y_{j} -\rho \right)  \nonumber\\ 
&  + \text{exp}\left(\text{log}\left(\overline{\Delta}-\underline{\delta}_{j}\right) - \frac{v_{j}}{w_{jj}} - \rho \right) \nonumber \\
&  + \text{exp} \biggl (\text{log}\left (c_{j}\right) - y_{j} - \rho \biggr ) \Biggr) \leq 0 \quad \forall j  \label{eq:C4}.
\end{align} 
\end{lemma}
\begin{proof}
We obtain \eqref{eq:C4} from \eqref{eq:p}/\eqref{eq:LP1}, which can be rewritten as $\sum^{n}_{i=1}p_{i}\left (A_{ij}-rI\right ) \leq -c_{j}$ for all $j$. Now using \eqref{eq:epi} this can be rewritten as 
\begin{equation}
\sum_{i\neq j}p_{i}\beta_{ij}-p_{j}\delta_{j}-p_{j}r \leq - c_{j},
\end{equation}
which is equivalent, when introducing $\overline{\Delta} > \overline{\delta_i}$, to 
\begin{equation}
\label{eq:posyC1}
\sum_{i \neq j} \frac{p_{i}\beta_{ij}}{p_{j}(\overline{\Delta}+r)} + \frac{\overline{\Delta}-\delta_{j}}{\overline{\Delta}+r}+ \frac{c_{j}}{p_{j}(\overline{\Delta}+r)} \leq 1 \quad \forall j.
\end{equation}
Taking the log of both sides and rewriting gives \eqref{eq:C4}
\end{proof}
 \begin{lemma}
The bounds on the spreading rate and recovery rate, respectively $0 < \underline{\beta}_{ij} \leq \beta_{ij} \leq \overline{\beta}_{ij}$ and $0 < \underline{\delta}_{i} \leq \delta_{i} \leq \overline{\delta}_{i}<\overline{\Delta}$, are equivalent to the following equations under convex transformation using $u_{ij}=f_{ij}(\beta_{ij})$ and $v_i=g_i(\overline{\Delta}-\delta_i)$ 
\begin{equation}
\label{eq:BC}
0 \leq u_{ij} \leq w_{ij}\text{log}\left (\frac{\overline{\beta}_{ij}}{ \underline{\beta}_{ij}}\right ), \quad 0 \leq v_{i} \leq w_{ii}\text{log}\left (\frac{\overline{\Delta}-\underline{\delta}_{i}}{ \overline{\Delta}-\overline{\delta}_{i}}\right )
\end{equation}
\end{lemma}
\begin{proof}
Derived in \cite{LCSS2021}.
\end{proof}
The bounds on the outbreak rate $\lambda_i$ and interval rate $\tau_i$ can be found in the same way.

\begin{table*}
\caption{Constraints of the minimization Problem \ref{P2} and their convex representation.}
\label{tab:const}
\centering
\begin{tabular}{ c|c | c} 
 Name & Model Constraint & Convex Representation   \\ 
\hline 
Dynamic coupling constraint & $p^{T}A - r p^{T} \leq - C$ & $q(y,u,v,\rho) \leq 0 \quad \forall j$ \\
Spreading rate bounds & $0 < \underline{\beta}_{ij} \leq \beta_{ij} \leq \overline{\beta}_{ij}$ & $0 \leq u_{ij} \leq \overline{u}_{ij}= w_{ij}\text{log}\left (\frac{\overline{\beta}_{ij}}{ \underline{\beta}_{ij}}\right )$\\
Recovery rate bounds & $0 < \underline{\delta}_{i} \leq \delta_{i} \leq \overline{\delta}_{i}<\overline{\Delta}$ & $0 \leq v_{i} \leq \overline{v}_i=w_{ii}\text{log}\left (\frac{\overline{\Delta}-\underline{\delta}_{i}}{ \overline{\Delta}-\overline{\delta}_{i}}\right )$ \\
Outbreak rate bounds & $0 < \underline{\lambda}_{i} \leq \lambda_{i} \leq \overline{\lambda}_{i}$ & $0 \leq z_{i} \leq \overline{z}_i=\omega_{\lambda_i}\text{log}\left (\frac{\overline{\lambda}_{i}}{ \underline{\lambda}_{i}}\right )$ \\
Revisit interval bounds & $0 < \underline{\tau}_{i} \leq \tau_{i} \leq \overline{\tau}_{i} $ & $0 \leq \sigma_{i} \leq \overline{\sigma}_i=\omega_{\tau_i}\text{log}\left (\frac{\overline{\tau}_{i}}{ \underline{\tau}_{i}}\right )$ \\
Discount rate bounds & $0\leq r \leq \overline{R} $ & $ \text{log}(\overline{\Delta}) \leq \rho \leq  \text{log}(\overline{\Delta}+\overline{R})$ \\
\end{tabular}
\end{table*}

Using the above objectives and constraints, we can now specify the standard set-up of the optimization problem
\begin{problem}
\label{P2}
Minimize the maximum risk bound $R_i=p_i \lambda_i \tau_i$ via sparse resource allocation, given defined resource allocation budgets $\Gamma$ and a cost $c_i$ associated with each node $i$. That is, find the optimal state matrix $A$, discount rate $r$, outbreak Poisson rate $\lambda_i$ and revisit interval $\tau_i$ that minimises 
\begin{align}
    \underset{p, r, \beta, \delta, \lambda, \tau}{\text{minimize}} & \quad \underset{i}{\text{max}}(p_i \lambda_i \tau_i)  \nonumber\\
     \text{such that} & \quad p\ge 0, \quad p^{T}A - r p^{T} \leq - C, \quad 0\leq r \leq \overline{R} \nonumber\\
& \quad 0 < \underline{\beta}_{ij} \leq \beta_{ij} \leq \overline{\beta}_{ij}, \quad 0 < \underline{\delta}_{i} \leq \delta_{i} \leq \overline{\delta}_{i} < \overline{\Delta} \nonumber \\
& \quad 0 < \underline{\lambda}_{i} \leq \lambda_{i} \leq \overline{\lambda}_{i}, \quad 0 < \underline{\tau}_{i} \leq \tau_{i} \leq 1 \nonumber\\
& \quad   \sum_{ij} f_{ij}\left( \beta_{ij}\right ) \leq \Gamma_\beta, \quad \sum_{i} g_{i}\left(\delta_{i}\right) \leq \Gamma_\delta  \nonumber \\
& \quad \sum_{i} h_{i}\left(\lambda_{i}\right) \leq \Gamma_\lambda, \quad \sum_{i}\psi_i \left(\tau_i\right) \leq \Gamma_\tau. \nonumber
\end{align} 
\end{problem}

\emph{Convex Problem 1:} The equivalent convex optimization problem, more specifically exponential cone program is  
\begin{align}
    \underset{y, \rho, u, v, z, \sigma}{\text{minimize}} & \quad \underset{i}{\text{max}}(y_i + \text{log}(\overline{\lambda}_i) - \frac{z_i}{\omega_{\lambda_i}}+ \text{log}(\overline{\tau}_i) - \frac{\sigma_i}{\omega_{\tau_i}})  \nonumber\\
    \text{such that} & \quad q(y,u,v,\rho) \leq 0, \quad \text{log}(\overline{\Delta}) \leq \rho \leq  \text{log}(\overline{\Delta}+\overline{R}) \nonumber \\
& \quad 0 \leq u_{ij} \leq \overline{u}_{ij}, \quad 0 \leq v_{i} \leq \overline{v}_i  \nonumber\\
& \quad 0 \leq z_{i} \leq \overline{z}_i, \quad 0 \leq \sigma_{i} \leq \overline{\sigma}_i\nonumber \\
& \quad  \sum_{ij} u_{ij}  \leq \Gamma_\beta, \quad \sum_{i} v_{i} \leq \Gamma_\delta  \nonumber\\
& \quad \sum_{i} z_i \leq \Gamma_\lambda, \quad \sum_i \sigma_i \leq \Gamma_\tau.  \nonumber
\end{align} 



\subsection{Sparsity and Reweighted $\ell_{1}$ minimization}
\label{subsec:L0}
A significant benefit of the the exponential cone programming formulation is that it encourages sparse resource allocation. Our resource models, e.g. $\sum_{ij} u_{ij}  \leq \Gamma_\beta$, are $\ell_{1}$ norm constraints and objectives, since $u_{ij}\geq 0$, which are widely used to encourage sparsity \cite{tibshirani1996regression,candes2006robust,donoho2006compressed,Candes2008}.

If the goal is maximal sparsity, i.e. minimal number of nodes and edges with non-zero resources allocated, then we can use the reweighted $\ell_{1}$ optimization approach of \cite{Candes2008}.  We can adapt this to our problem by iteratively solving the problem, but with a reweighted resource model that approximates the number of nodes and edges with non-zero allocation, e.g.
\begin{equation}
\label{eq:L1} 
\phi^k=\sum_{ij} \frac{u^{k}_{ij}}{u_{ij}^{k-1}+\epsilon}
\end{equation} 
where $k$ is the iteration number and $\epsilon$ a small positive constant to improve numerical stability. If we use our resource model as a constraint in the optimization problem, we now replace the constraint with $\phi^k \leq M$ where $M$ is the bound on the number of nodes and edges that can have resources allocated to them. This iteration has no guarantee of convergence or global optimality, but has been found to be very effective in practice.

\section{Results}
In this section we illustrate how our proposed method can be implemented with four different examples. First of all, we look into the effect of different parameters on the risk with a small 16 node example. Next, we discuss a 7 node epidemic example to demonstrate how the optimization framework can be used for identifying critical links and deriving a vaccination strategy. Third, we discuss a larger 359 nodes epidemic example of a pandemic spreading on the airport transportation network in the US. Finally, we discuss a large wildfire example of 4000 nodes, demonstrating how we can utilize resource allocation and revisit schedules to minimize the risk of an undetected outbreak and incorporate robot path planning methods for surveillance. 

\subsection{Risk Map - 16 nodes}
Let us consider a graph with $n=16$ nodes, connected as visualized in Fig. \ref{fig:16parameters}. We now want to investigate the effect of the different parameters, i.e. cost, outbreak rate and spreading rate, on the risk. We first set all parameters the same for all nodes and links, i.e. we take $c_{i}=1$, $\beta_{ij}=0.5$ and $\lambda_i=1$ and then introduce one by one node dependent cost, spreading rate and outbreak rate. The resulting risk maps for $r=2$ are displayed in Fig. \ref{fig:16parameters}.  

\begin{figure}[!ht]
\centering
    \begin{subfigure}[b]{\linewidth}
\centering   
              \def\svgwidth{0.48\textwidth}
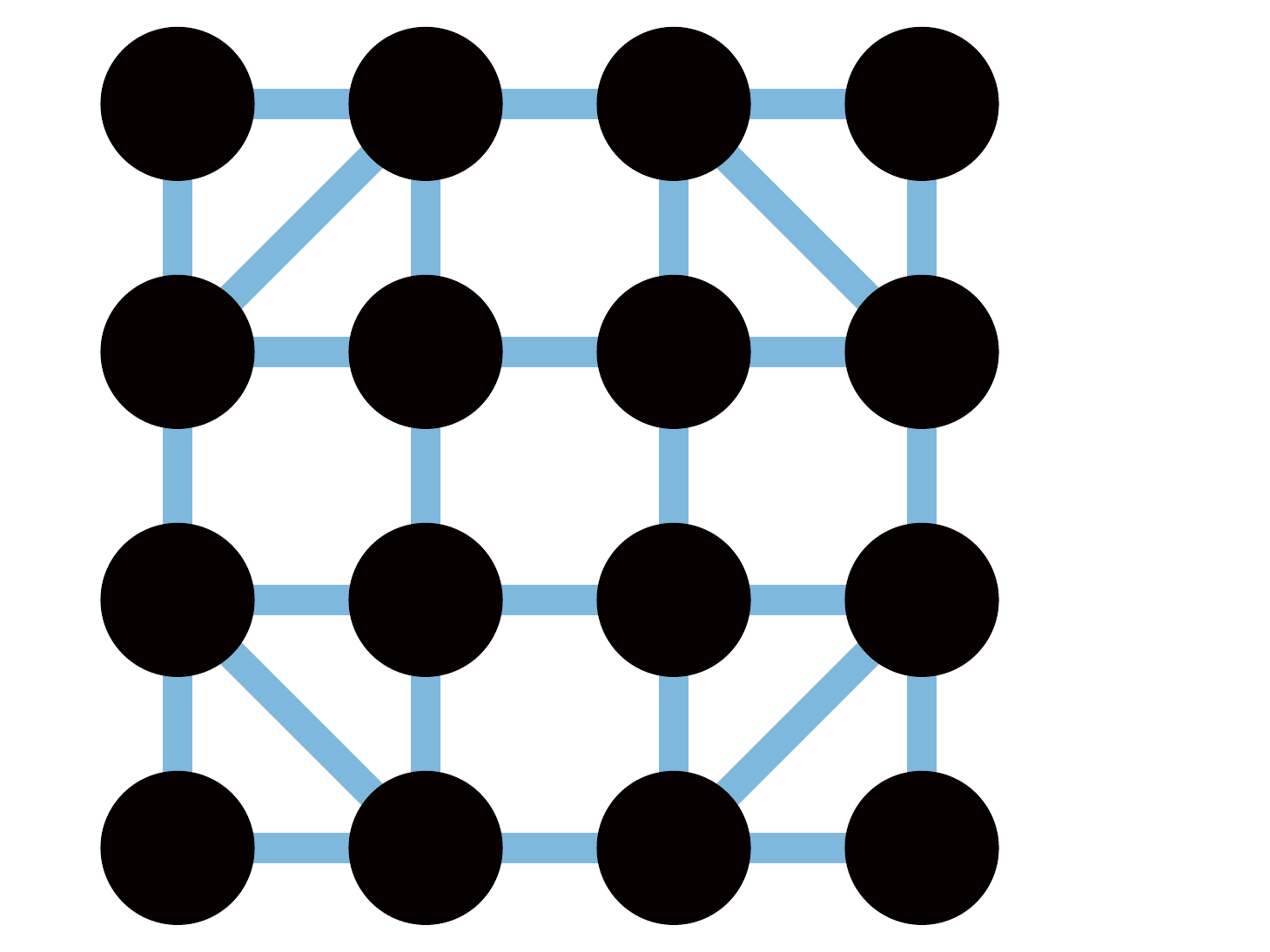     
             \hfill    
	\def\svgwidth{0.48\textwidth}
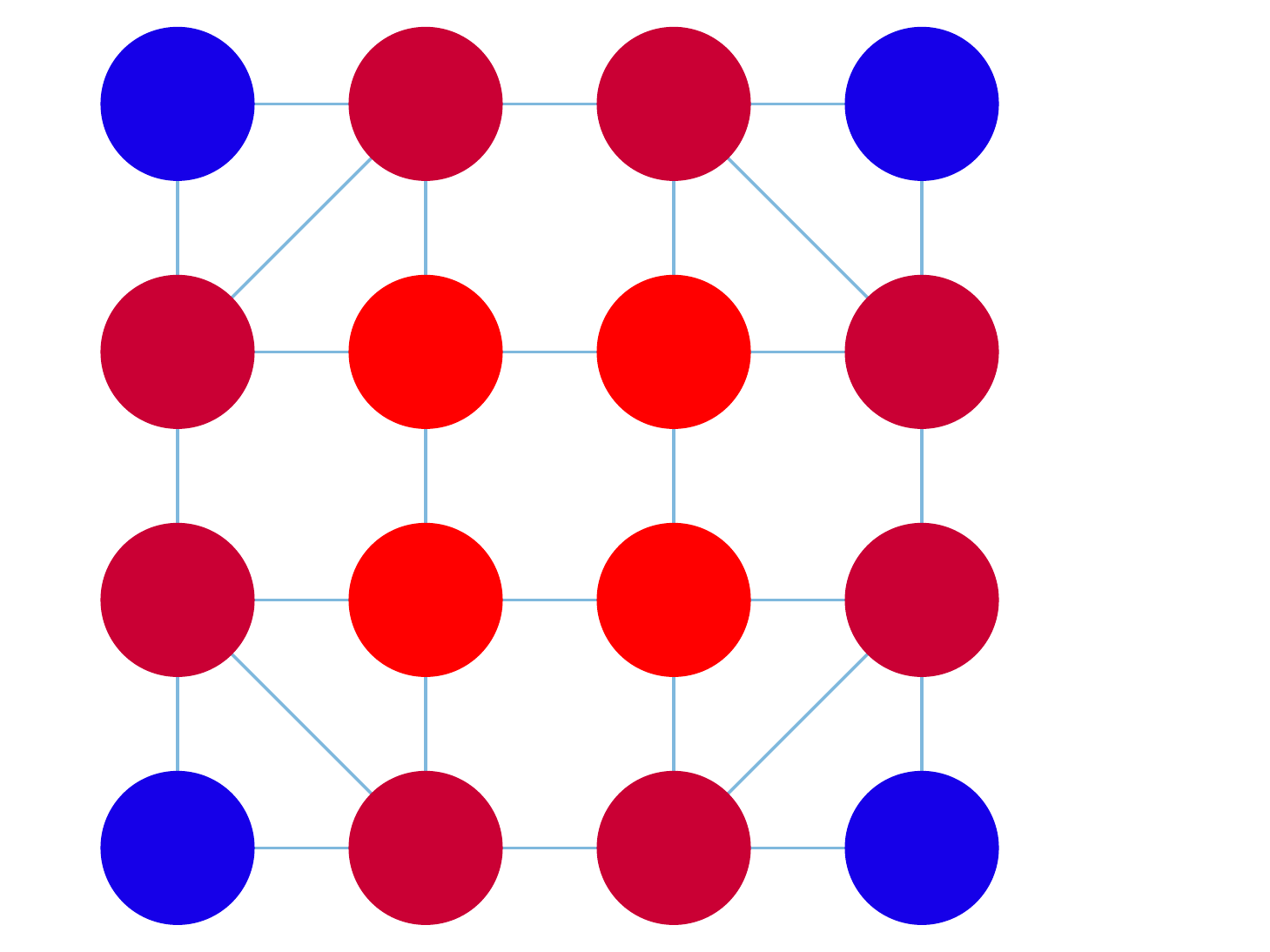
            \caption{All nodes have equal cost $c_i=1$ and outbreak rate $\lambda_i=1$, all edges have spreading rate $\beta_{ij}=0.5$}
            \label{fig:16p11}
    \end{subfigure}%
\vskip\baselineskip
    \begin{subfigure}[b]{\linewidth}    
            \centering        
 \def\svgwidth{0.48\textwidth}
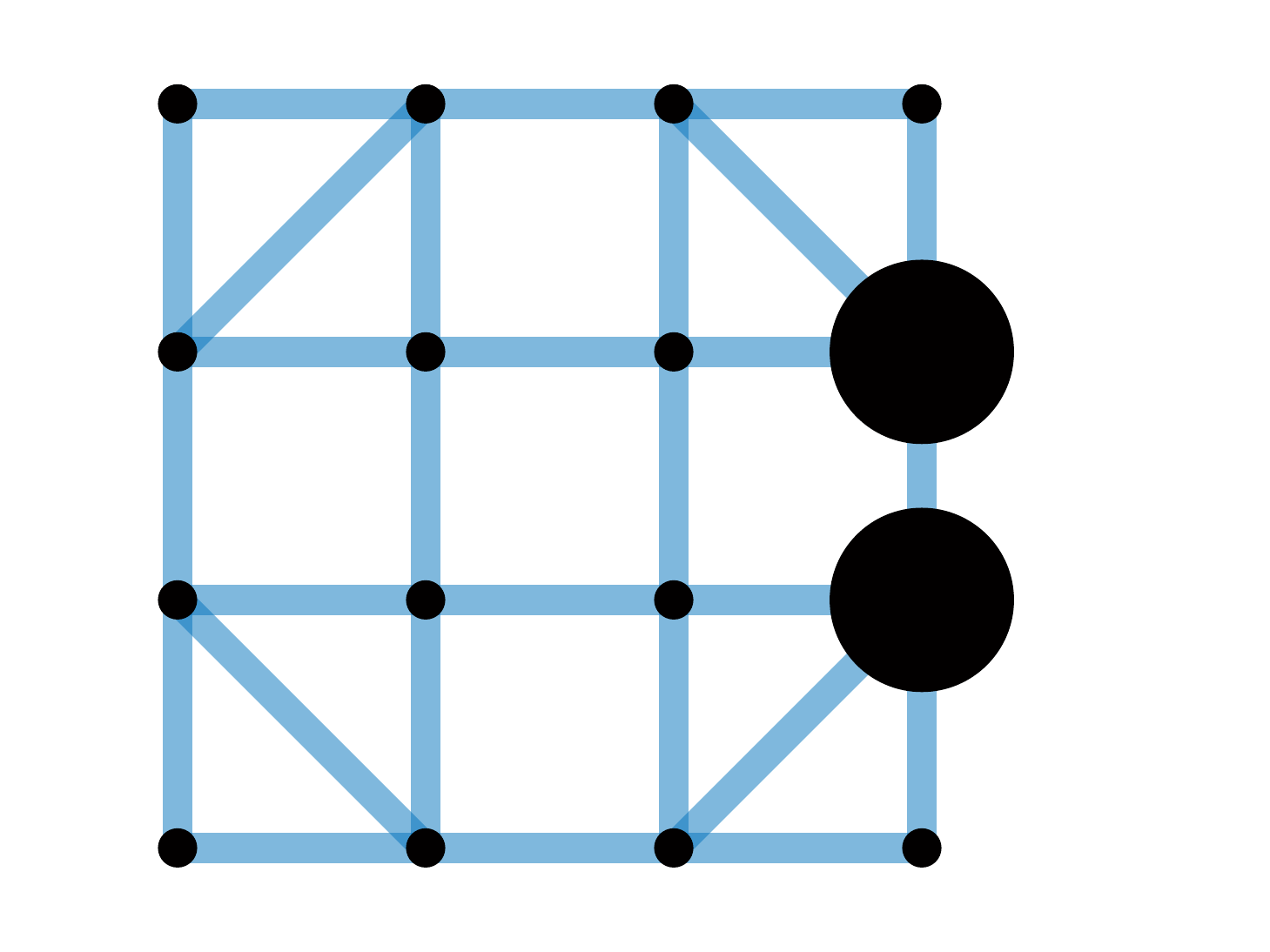     
\hfill
\def\svgwidth{0.48\textwidth}
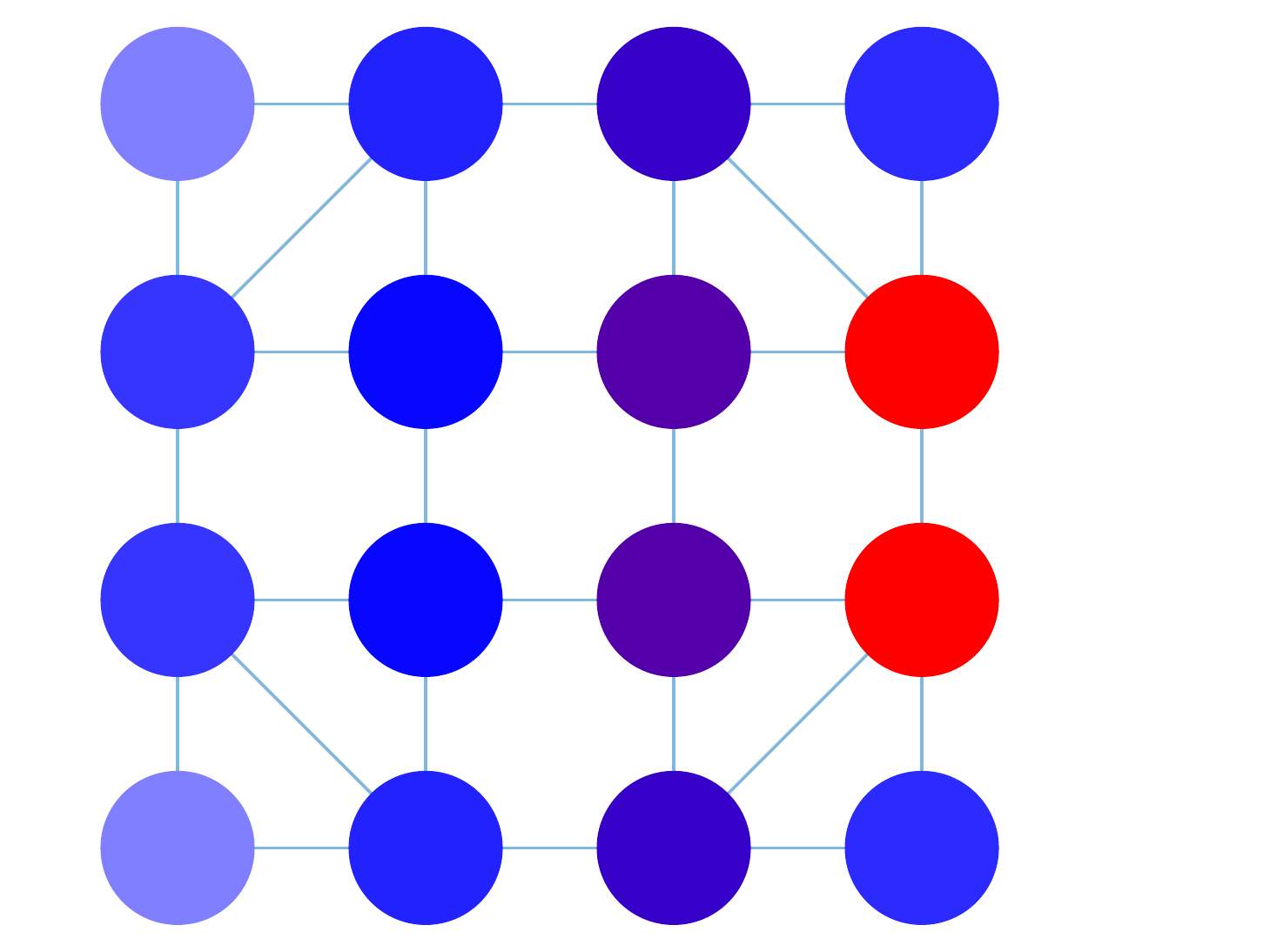
            \caption{High cost nodes $i=14,15$, equal outbreak rate $\lambda_i=1$ for all nodes and equal spreading rate $\beta_{ij}=0.5$ for all edges }
            \label{fig:16p33}
    \end{subfigure}%
\vskip\baselineskip
    \begin{subfigure}[b]{\linewidth}
            \centering
          \def\svgwidth{0.48\textwidth}
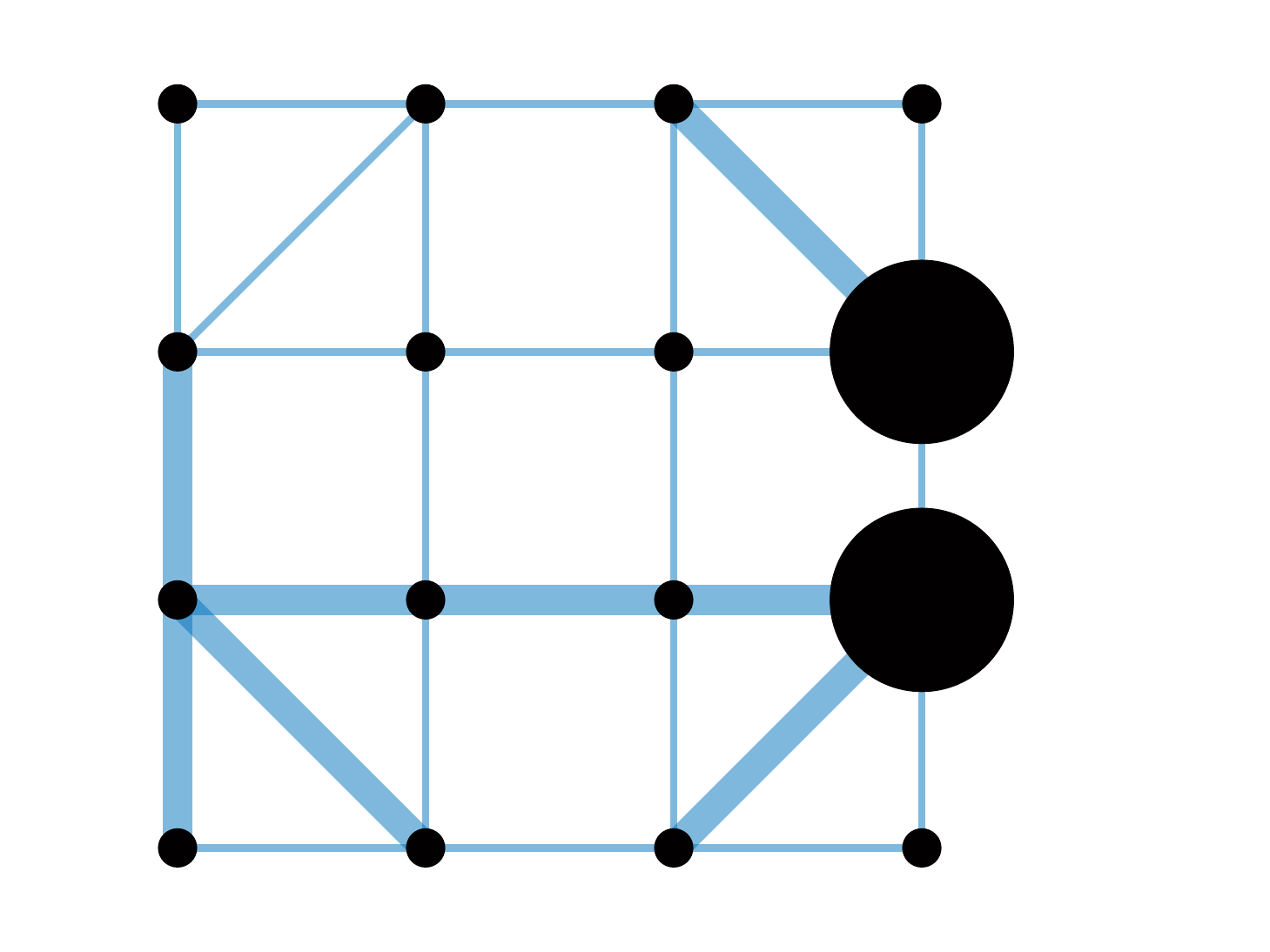     
\hfill
\def\svgwidth{0.48\textwidth}
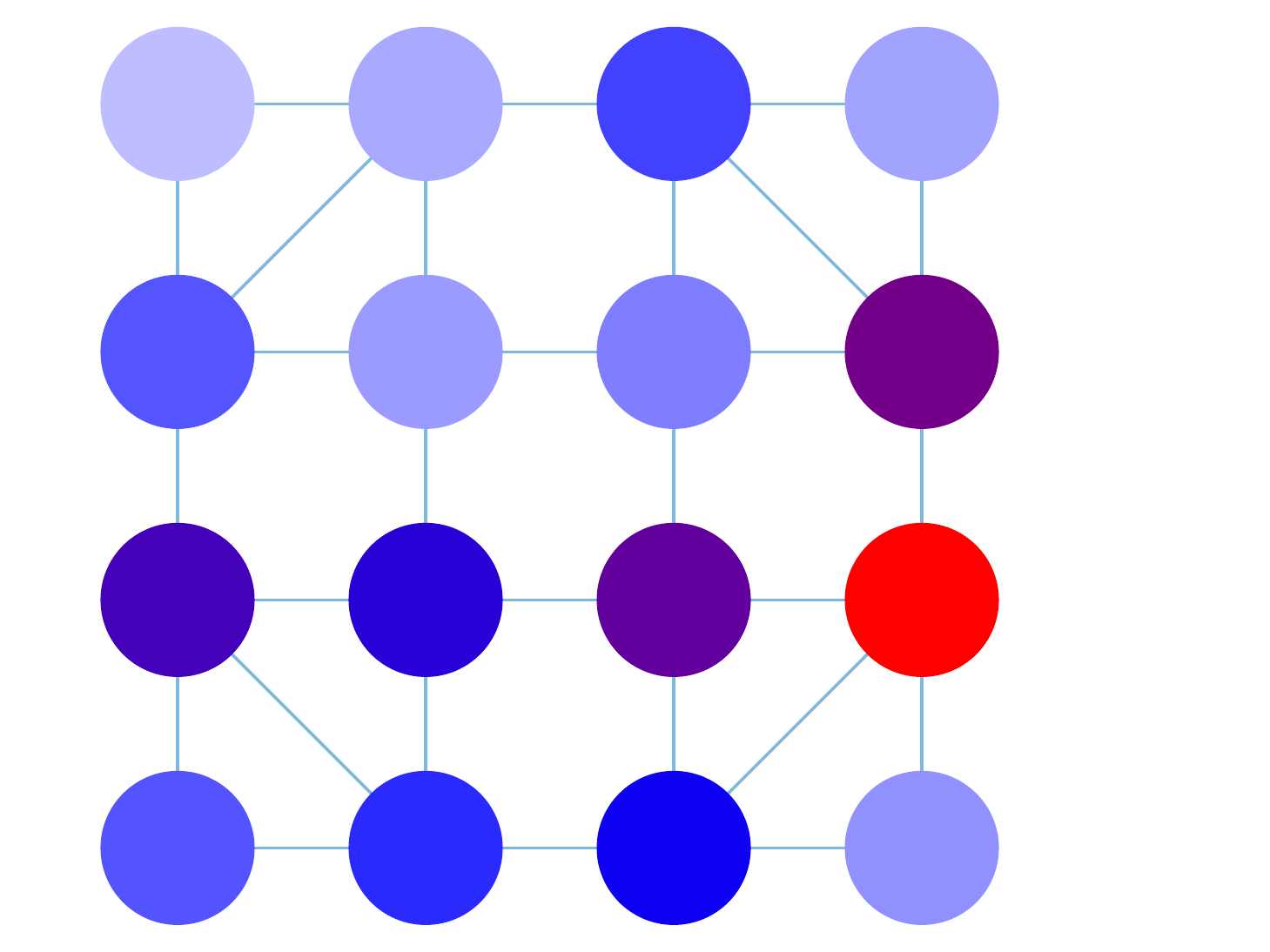
            \caption{Spreading rate as indicated by edge width on the left, high cost nodes $i=14,15$, equal outbreak rate $\lambda_i=1$}
            \label{fig:16p66}
    \end{subfigure}
\vskip\baselineskip
     \begin{subfigure}[b]{\linewidth}
            \centering
\def\svgwidth{0.48\textwidth}
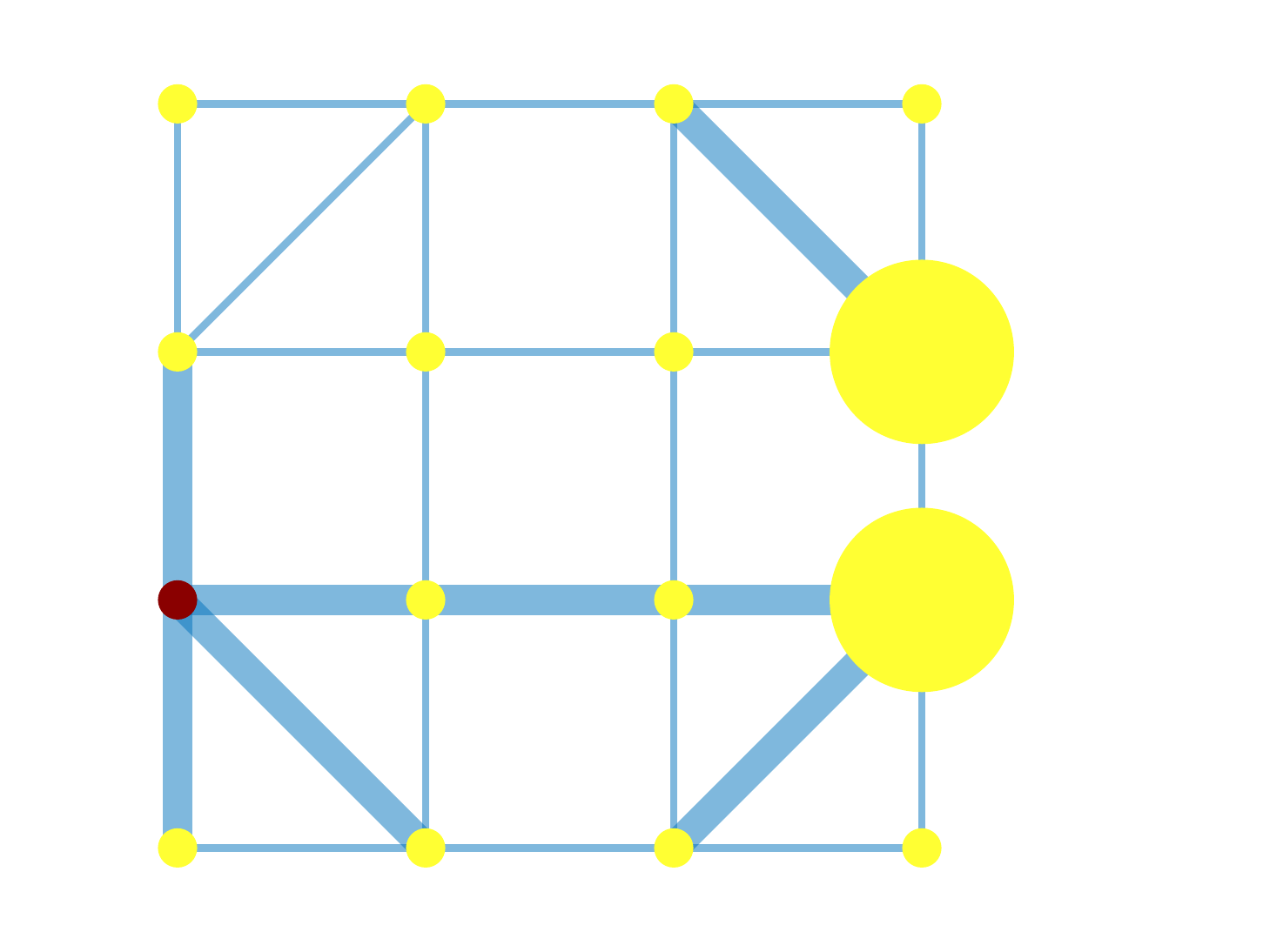     
\hfill
 \def\svgwidth{0.48\textwidth}
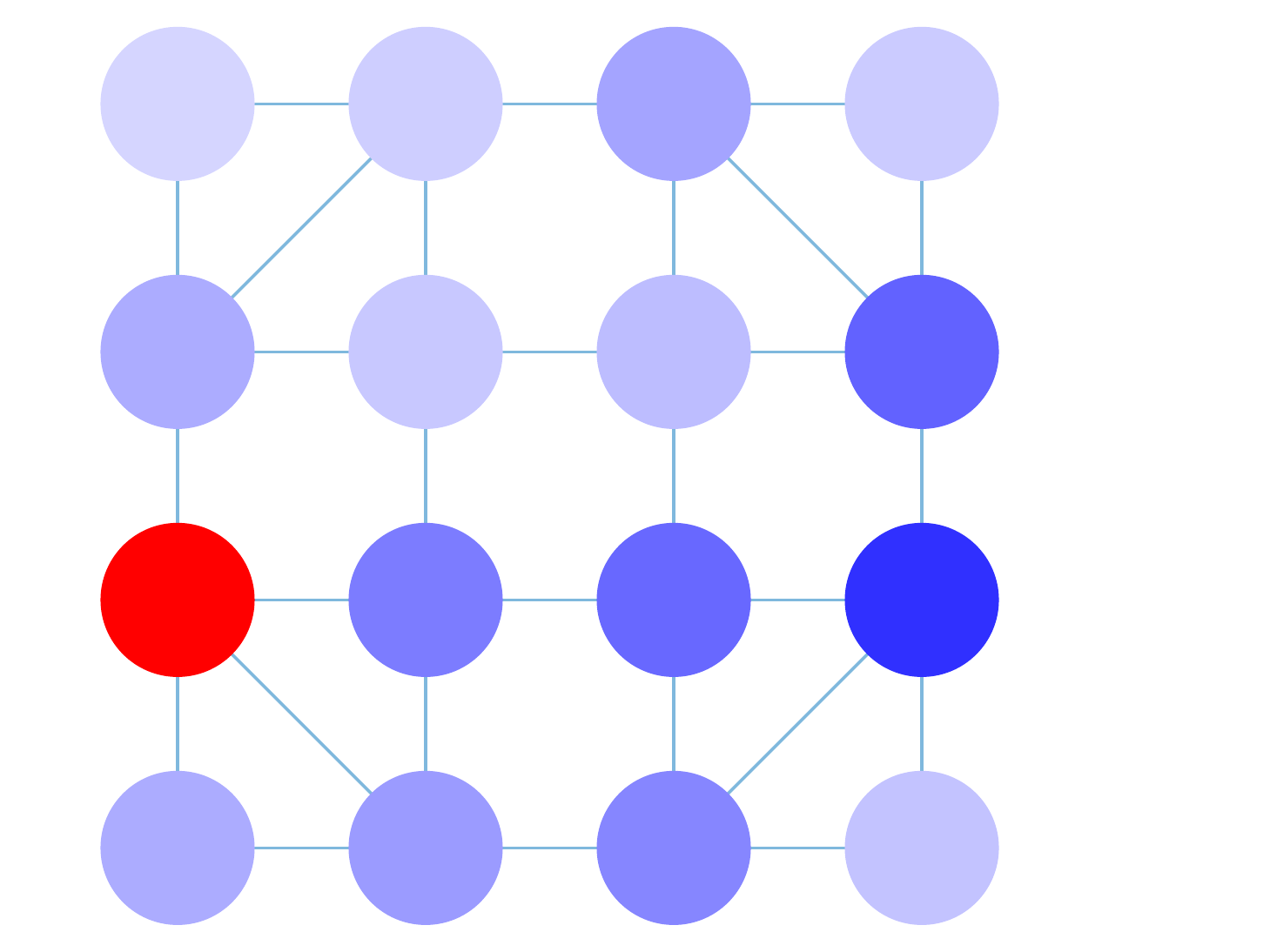
            \caption{Outbreak rate, spreading rate and cost as indicated on the left}
            \label{fig:16p88}
    \end{subfigure}
    \caption{Example 1; The effect of cost, outbreak rate and spreading rate on a graph with $n=16$ nodes. On the left is the graph where node colour indicates outbreak rate, node marker size indicates cost and edge width indicates spreading rate. The resulting risk map is displayed on the right, where node colour indicates risk.}
\label{fig:16parameters}
\end{figure}

It can be seen that uniform spreading rate, cost and outbreak rate result in risk based on node degree, see  Fig. \ref{fig:16p11}.

If we introduce a cost $c_{i}=1$ for nodes $i=14,15$ and $c_{i}=0.1$ otherwise, see Fig. \ref{fig:16p33}, we obtain high risk nodes $i=14,15$ which correspond to the high cost nodes. 

Including higher spreading rates $\beta_{ij}=0.8$ at specified edges and a lower spreading rate of $\beta_{ij}=0.2$ elsewhere, see Fig. \ref{fig:16p66}, we obtain high risk node $i=15$, which corresponds to both high cost and high spreading rate. 

Finally, we set outbreak rate as indicated by the colourbar with a high likelihood of an outbreak starting in node $i=3$. It can be seen in Fig. \ref{fig:16p88} that the high risk node $i=3$ now corresponds to the high outbreak node. 

We can visualize the effect of the discount rate for Fig. \ref{fig:16p88} by comparing the resulting risk map for a low and high discount rate, see Fig. \ref{fig:Small16dr}. Here the minimum discount rate equals $r=1.6439$ before $rI-A$ becomes singular. A higher discount rate, see Fig. \ref{fig:16highdr}, prioritizes the near future and hence, prioritizes what is happening in the high cost nodes over what possible outbreaks in other parts of the graph could lead to. 

\begin{figure}
\centering
\begin{subfigure}[b]{0.48\linewidth}
        \def\svgwidth{1\textwidth}
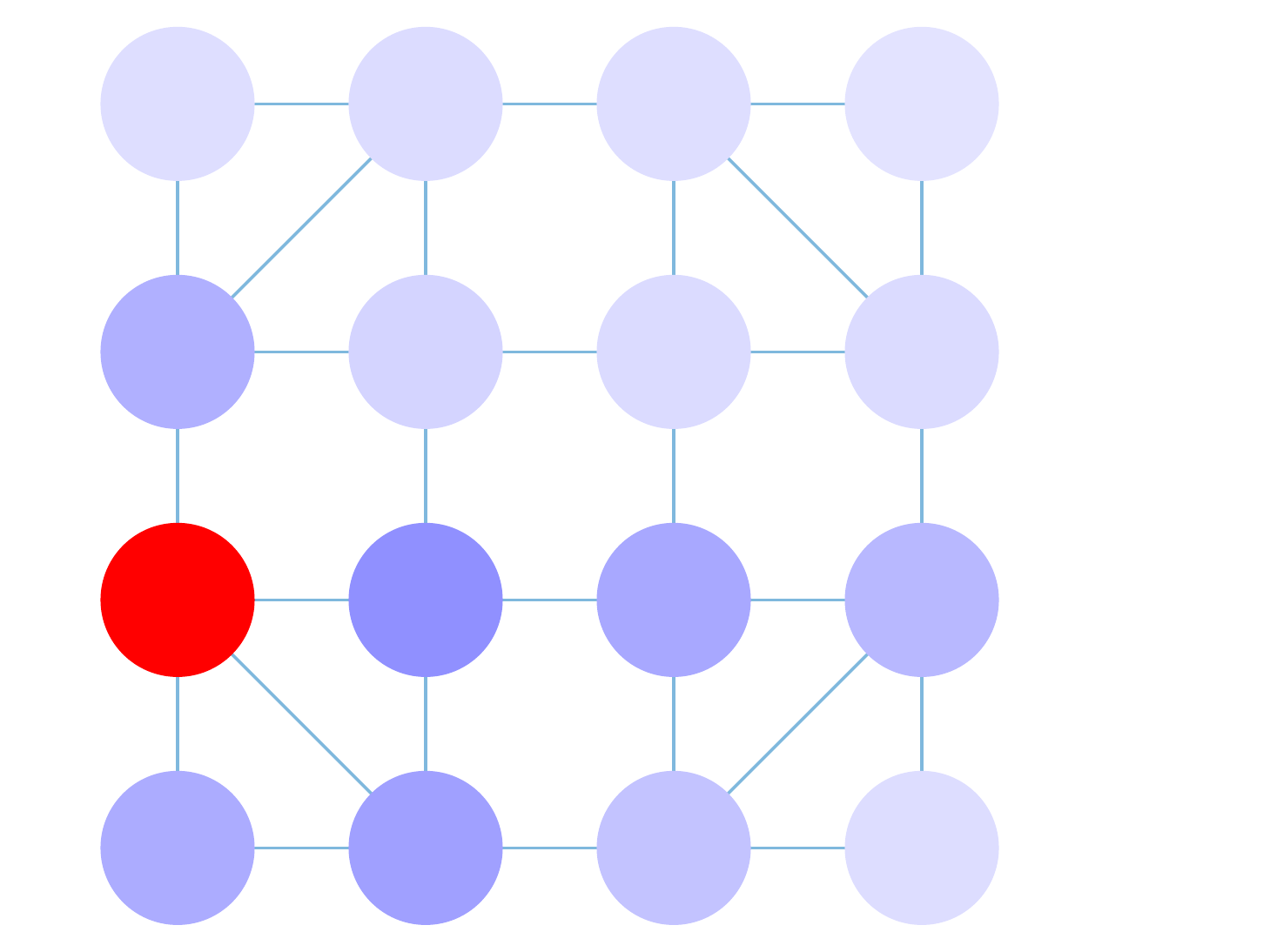
      \caption{$r=1.6439$}
      \label{fig:16lowdr}
  \end{subfigure}
  ~ 
     \begin{subfigure}[b]{0.48\linewidth}
\def\svgwidth{1\textwidth}
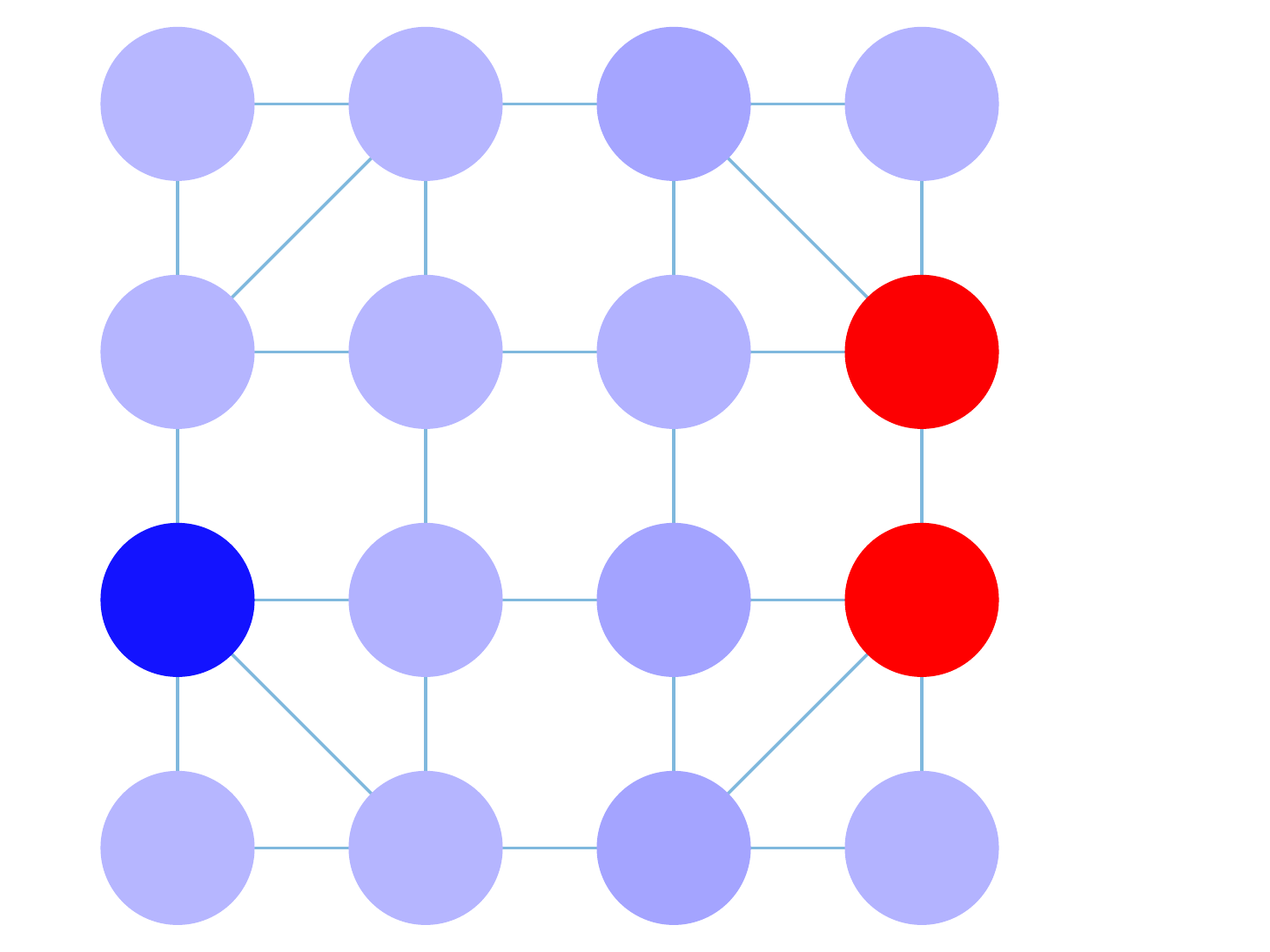
      \caption{$r=20$}
      \label{fig:16highdr}
  \end{subfigure}
      \caption{Effect of the discount rate on the risk map of Fig. \ref{fig:16p88}.}
\label{fig:Small16dr}
\end{figure}

\subsection{Epidemic Intervention - 7 nodes}
In the previous example we discussed the risk map. The next step is the optimal control problem of minimizing risk by modifying the systems dynamics as explained in Section \ref{sec:CO}. To demonstrate how the optimization framework can be utilized we take an example of an epidemic spreading on a graph with $n=7$ nodes as visualized in Fig. \ref{fig:GraphAM1}. Here, node colour indicates outbreak rate, weightings $w$ are indicated at the edges and node marker size indicates cost. In this particular example we have a vulnerable high cost node $i=7$, which could be an elderly person, which is connected to $i=6$, which can be thought of as a caretaker or nurse. This link has therefore a higher weight $w_{ij}=10$, hence a larger cost to remove it. Node $i=6$ is connected to node $i=1$ who has a high likelihood of getting the disease, e.g. due to their work situation.

\begin{figure}
\centering
 \def\svgwidth{0.85\linewidth}
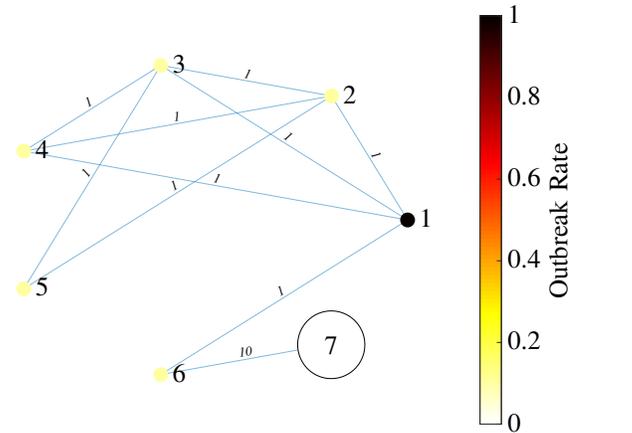    
\caption{Example 2; Epidemic spreading over a graph with $n=7$ nodes. Node colour indicates outbreak rate $\lambda_i$, weightings $w_{ij}$ are indicated at the edges, node marker size indicates cost $c_i$.}
\label{fig:GraphAM1}
\end{figure}

We are, first of all, interested to see if our method can indicate the critical link between node $1$ and $6$. Hence, we take our objective as minimizing the risk \eqref{eq:Rlamb}, while being able to allocate resources on the links, i.e. spreading rate, take $\Gamma_\beta=1$ and obtain Fig. \ref{fig:AM1}. Our method does indeed suggest allocating resources to this particular link, especially considering the larger cost involved in breaking up the link between the nurse and the associated patient. We compare our method to minimizing the spectral radius \eqref{eq:spec}, see Fig. \ref{fig:AM1S}. It can be seen, that by only minimizing the spectral radius instead of our proposed risk model, the same critical link is not identified. 

\begin{figure}[!ht]
\centering
              \def\svgwidth{0.85\linewidth}
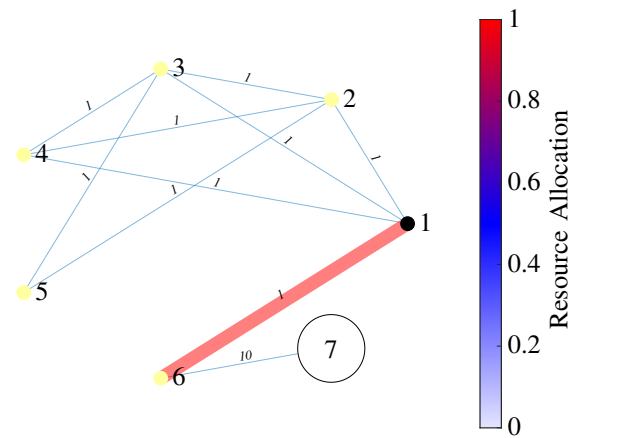   
\caption{Resource allocation for minimizing the risk \eqref{eq:Rlamb}. Node colour indicates outbreak rate $\lambda_i$, weightings $w_{ij}$ are indicated at the edges, node marker size indicates cost $c_i$.}
\label{fig:AM1}
\end{figure}

\begin{figure}
\centering
              \def\svgwidth{0.85\linewidth}
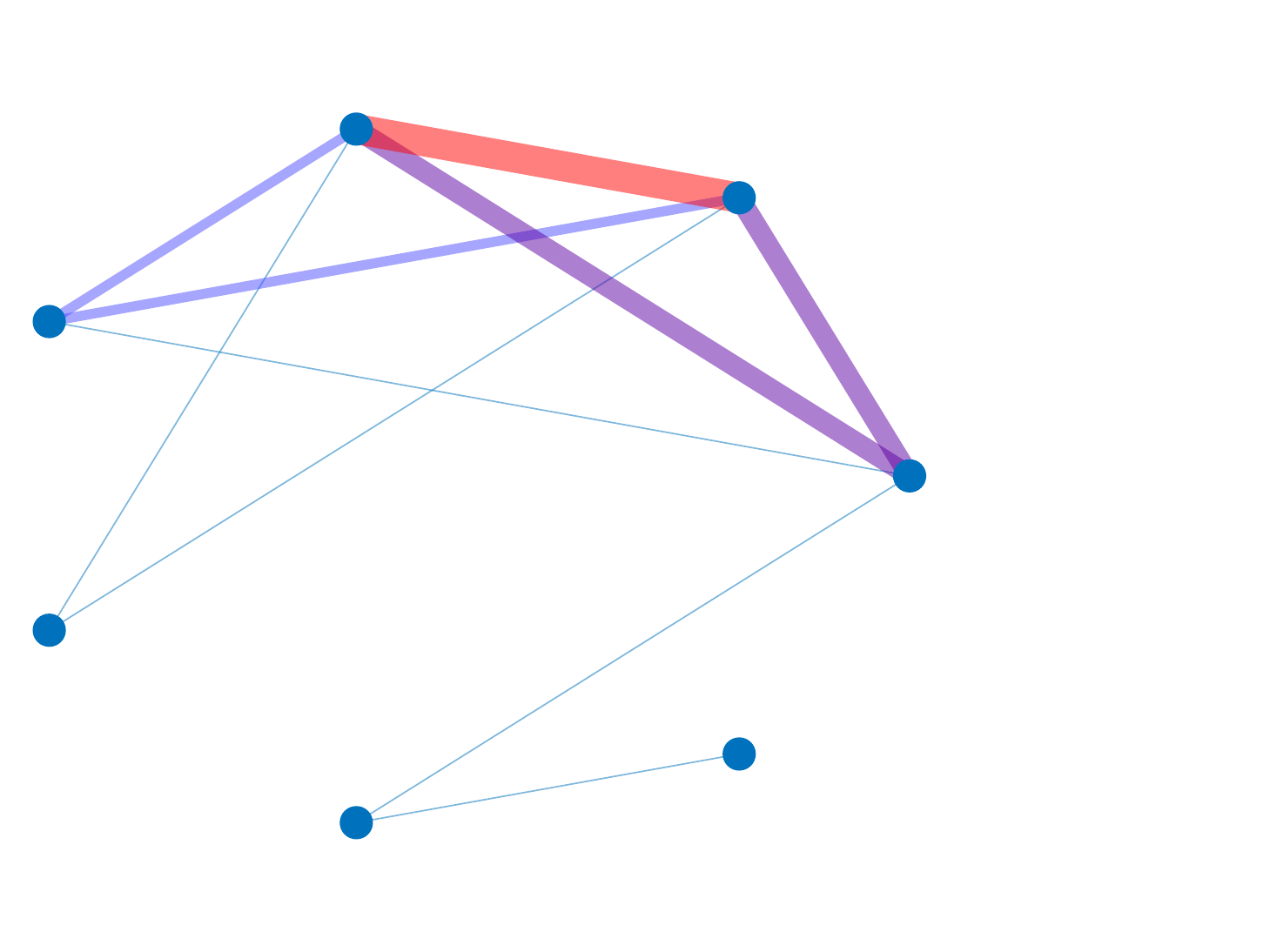   
\caption{Resource allocation for minimizing the spectral radius \eqref{eq:spec}.}
\label{fig:AM1S}
\end{figure}

Another common scenario in relation to epidemics, is vaccination. We model vaccination in such a way that similar to other common methods vaccination of a node increases both its recovery rate \cite{forster2007optimizing,hansen2011optimal,Han2015} and all incoming spreading rates \cite{zaman2008stability,hethcote2000mathematics,chen2006susceptible,kar2011stability}. Setting $\overline{\Delta}=1$ the obtained resource allocation, or vaccination strategy is given in Fig. \ref{fig:7BDcVac}. The initial zero outbreak rate of node 7 and higher weight on the edge between node 6 and 7, make in this case vaccination of node 7 not relevant. 

\begin{figure}
\centering
              \def\svgwidth{0.85\linewidth}
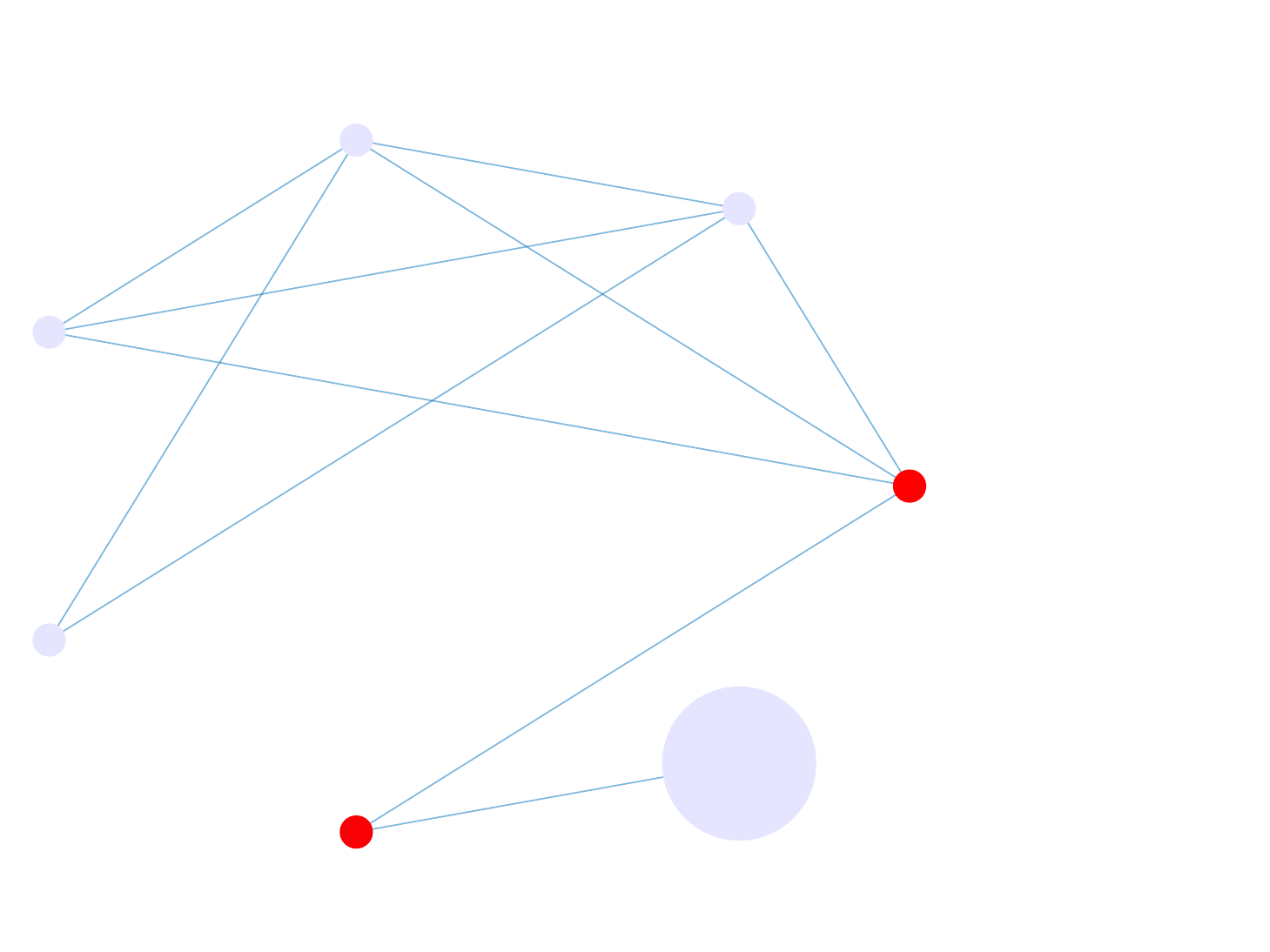   
\caption{Resource allocation for vaccination scenario. Node colour indicates nodes vaccinated, i.e. $\beta_{ij}$ is reduced and $\delta_i$ is increased for nodes $i=1$ and $i=6$, $r=1.3$.}
\label{fig:7BDcVac}
\end{figure}
%

\subsection{Air Transportation Network}
We now demonstrate resource allocation and policy making in regards to risk management for a larger network given a spreading process, here an epidemic. We consider the domestic US air transportation network graph based on the amount of passengers transported in 2014 \cite{USdata} consisting of $n=359$ nodes and 2097 edges as visualized in Fig. \ref{fig:AirG}. All domestic air traffic within the US excluding Puerto Rico, Virgin Islands and Guam is considered, where multiple airports that serve the same city are combined and routes that have less than 10.000 pax are omitted.  

\begin{figure}
\centering
\includegraphics[width=0.9\linewidth]{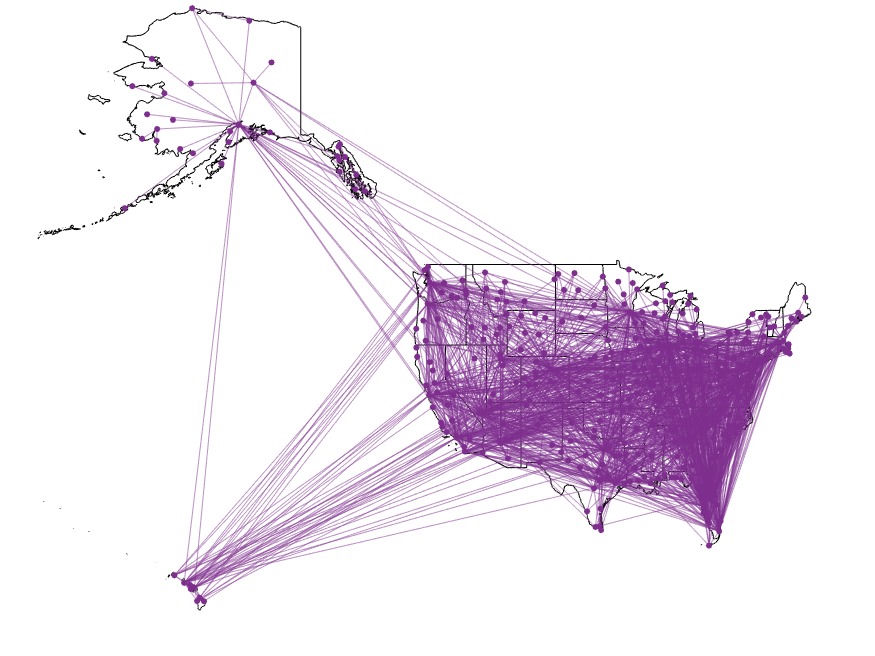}
\caption{Example 3; Epidemic spreading over a graph of the domestic US air transportation network with 359 nodes and 2097 edges based on \cite{USdata}.}
\label{fig:AirG}
\end{figure}

Now considering $\beta$ allocation only and the spread of SARS, we obtain a spreading rate of $\beta=0.25$ per day, a recovery rate of $0.0352$ and a death rate of $0.0279$ \cite{chowell2004model}. For simplicity we consider recovered nodes as removed as well due to (temporary) immunity. For the resource model we now base the weighting $w_{ij}$ on the number of passenger transported, taking the logic that the more flights and passengers, the higher the cost of reducing this traffic. Our objective is now to minimize the amount of risk by 50 percent, given an outbreak in Philadelphia (PHL) ($\hat{x}(0)_{i}=1$). The cost $c_i$ equals the normalized number of passengers served per airport and $r=10.7$. Note that this is a variation of the proposed Problem \ref{P2}, namely risk-constrained resource minimization. The resulting risk map is given by Fig. \ref{fig:AirRA} and reduces passenger numbers on 53 out of 2097 routes.

\begin{figure}
\centering
\def\svgwidth{0.95\linewidth}
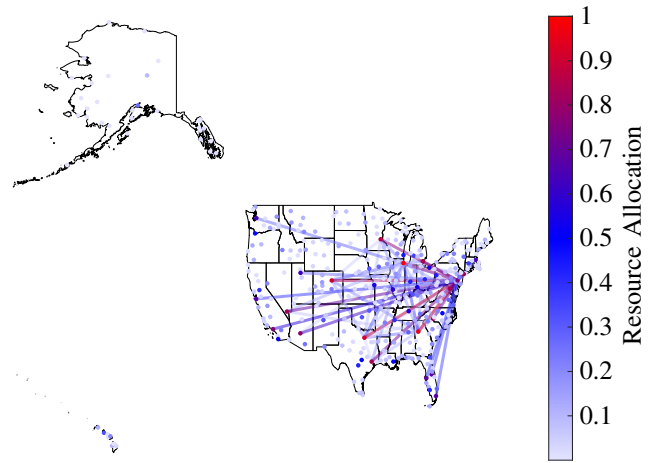   
\caption{Resource allocation given known outbreak in PHL for reducing the risk by 50\%. 53 edges have resources allocated. }
\label{fig:AirRA}
\end{figure}

We can further improve sparsity by solving the weighted $\ell_{1}$ minimization as explained in Section \ref{subsec:L0}. The resulting allocation is demonstrated in Fig. \ref{fig:AirL1}, here the number of affected routes is reduced to only 11 instead of 53, while achieving the same risk as in Fig. \ref{fig:AirRA}. Hence, by stopping air traffic between Philadelphia and respectively Atlanta, Charlotte, Chicago, Dallas, Denver, Detroit, Houston, Las Vegas, Minneapolis/St. Paul, New York City and Washington DC, risk can be reduced by 50 percent while leaving other busy air routes intact. 

\begin{figure}
\centering
\def\svgwidth{0.95\linewidth}
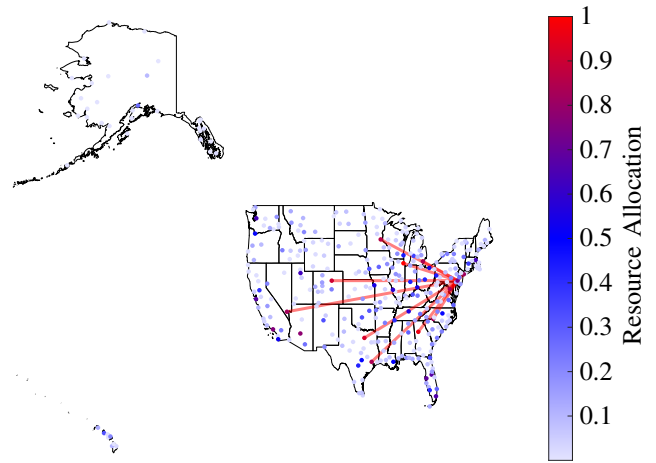   
\caption{Resource allocation after using reweighted $\ell_1$ minimization on Fig. \ref{fig:AirRA}. 11 edges have resources allocated.}
\label{fig:AirL1}
\end{figure}


\subsection{Wildfire Example}
Finally we consider a wildfire example to demonstrate risk maps, revisit maps and resource allocation in combination with persistent monitoring. 

Let us consider a fictional landscape with varying vegetation, a city and water as visualized in Fig. \ref{fig:NL4000V}. It can be represented as a network graph with $n=4000$ nodes. The adjacency matrix and therefore, the set of edges $\mathcal{E}$, is based on a 8 node spreading direction grid, i.e., the landscape is seen as a grid where each node is connected to its direct horizontal, vertical and diagonal neighbors. The recovery rate $\delta=0.5$ for all nodes, whereas $\beta$ is generated using data from cellular automata wildfire models presented in \cite{Karafyllidis1997a} and \cite{Alexandridis2008a}, where the wildfire spreading probabilities are determined based on real wildfire observations. For the spreading dynamics $\beta_{veg}=2 \cdot \text{Vegetation}$, where the vegetation value is indicated in Fig. \ref{fig:NL4000V} and $\in [0,1]$ for respectively minimal and maximal vegetation. For the city $\beta_{veg}=0.5$ is taken and for the nonburnable water areas $\beta_{veg}=0$.  Finally, $\beta$ is corrected for spreading between diagonally connected nodes, following \cite{Karafyllidis1997a}. The cost $c_i$ and outbreak rate $\lambda_i$ are given in respectively Fig. \ref{fig:NL4000C} and Fig. \ref{fig:NL4000L}.

\begin{figure}
\centering
\begin{subfigure}[b]{0.85\linewidth}
\def\svgwidth{1\textwidth}
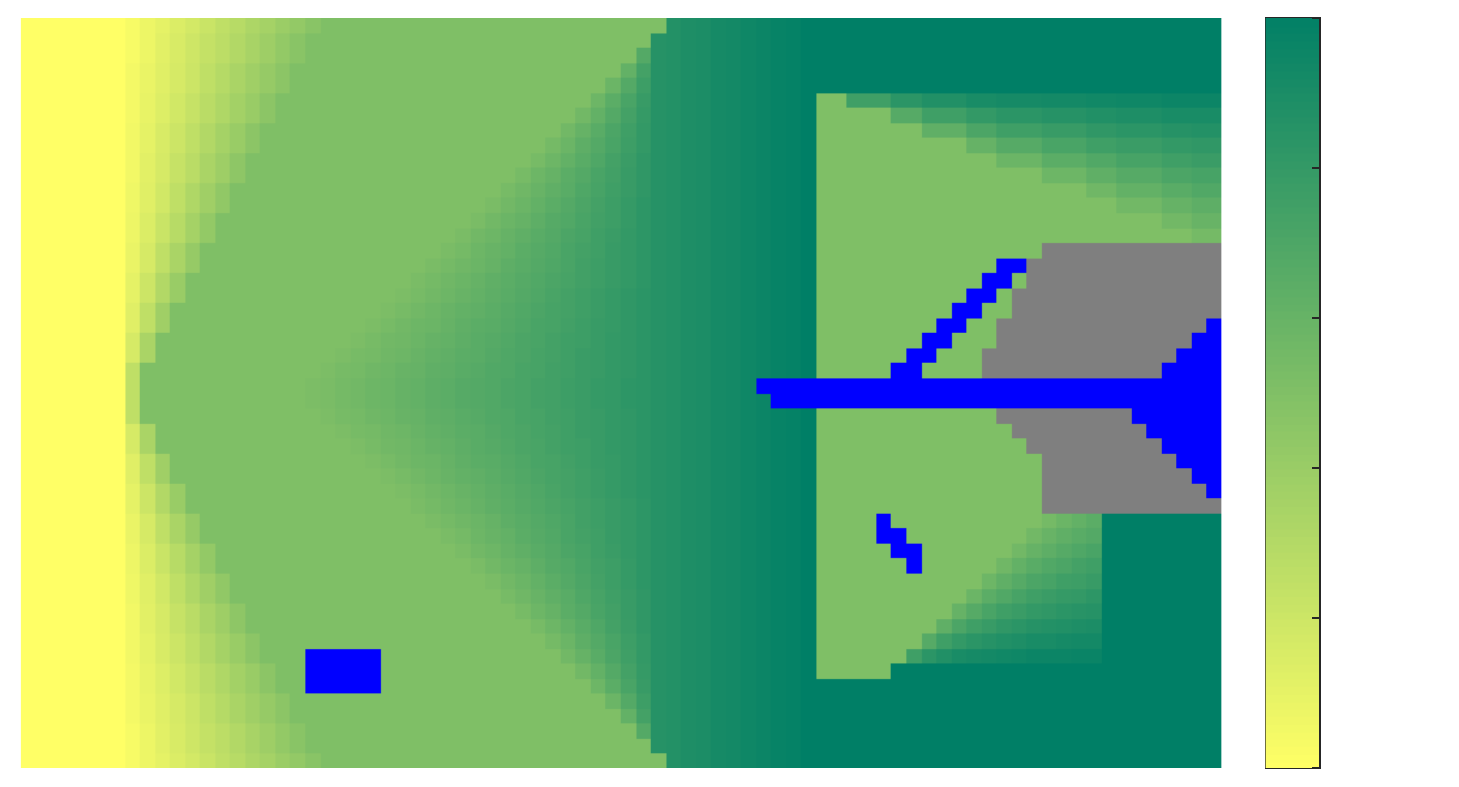   
      \caption{Vegetation Map; blue indicates water, grey city area.}
\label{fig:NL4000V}
\end{subfigure}
\begin{subfigure}[b]{0.48\linewidth}
\def\svgwidth{1\textwidth}
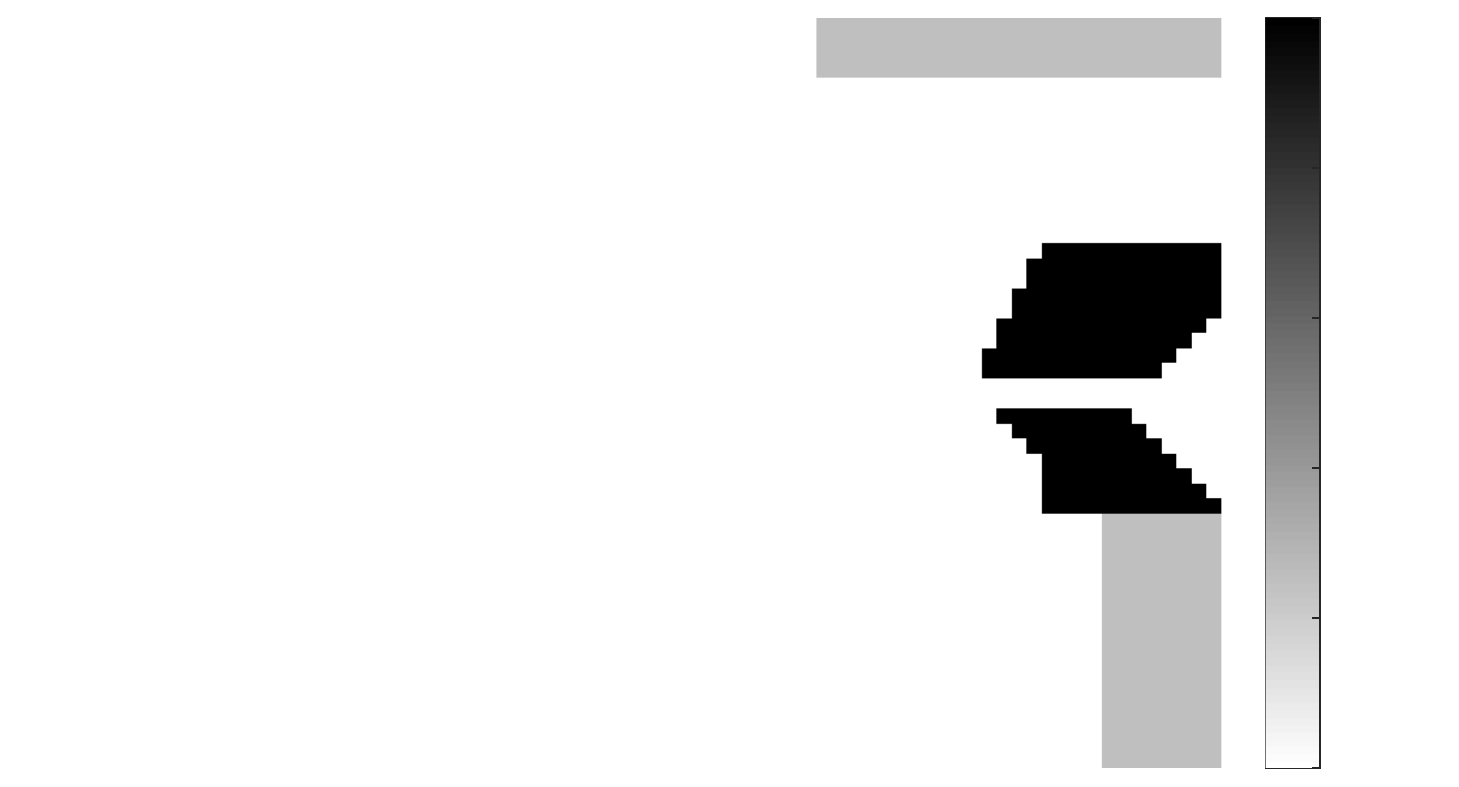   
      \caption{Cost Map}
      \label{fig:NL4000C}
  \end{subfigure}
  ~ 
     \begin{subfigure}[b]{0.48\linewidth}
\def\svgwidth{1\textwidth}
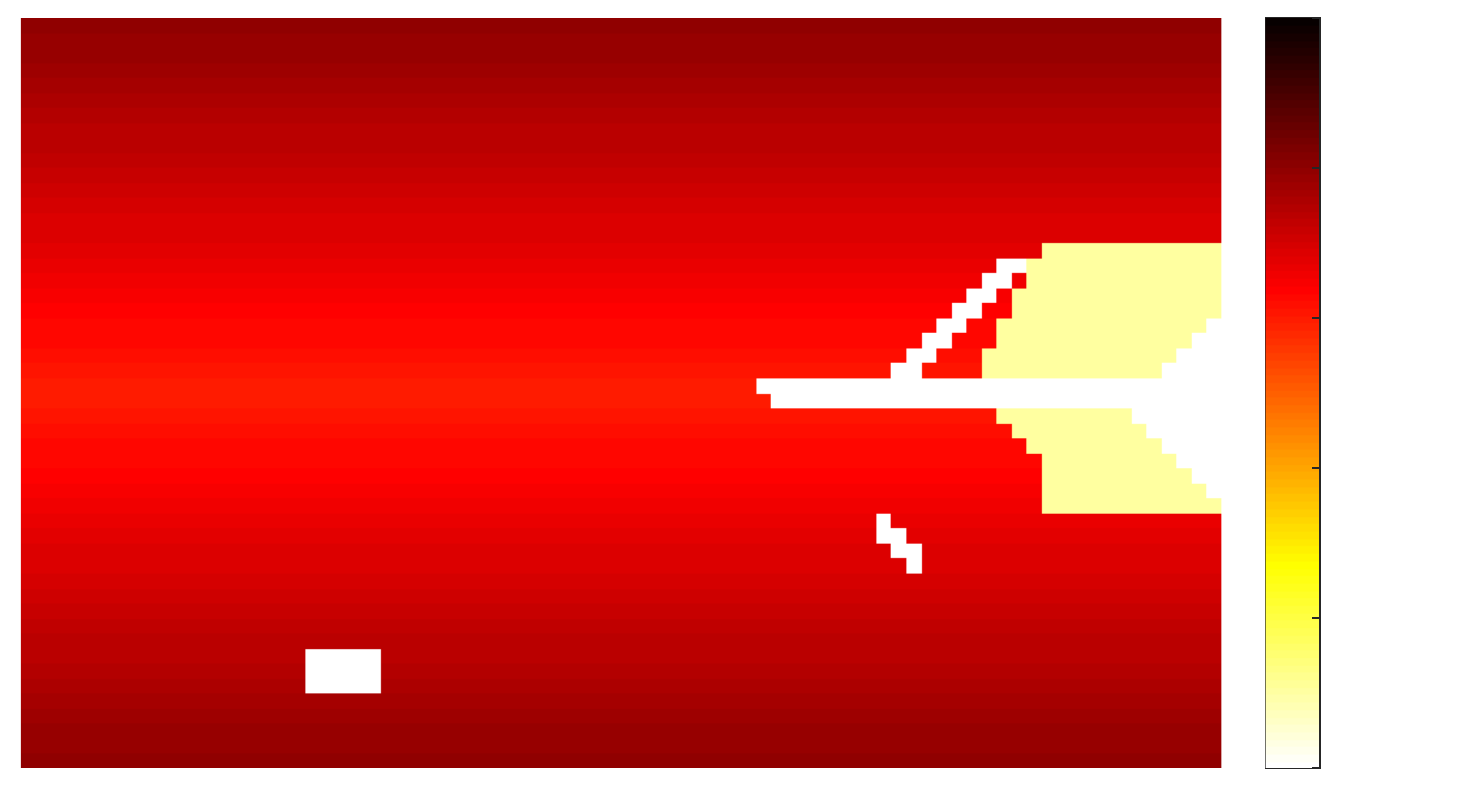   
      \caption{Outbreak Rate}
      \label{fig:NL4000L}
  \end{subfigure}
      \caption{Example 4; Wildfire spreading over a graph with $n=4000$ nodes representing a simplified landscape with its corresponding cost map and outbreak rate.}
\label{fig:NL4000}
\end{figure}

The spreading rate is furthermore adjusted for respectively a westerly, northerly, easterly and southerly wind with a speed of $V=8$ m/s. Taking a discount rate of $r=4$, this results in Fig. \ref{fig:4000Wind2}, where the effects of wind direction on the risk map can be seen. Large spreading rates, caused by e.g. eucalyptus forest, in combination with wind direction form the most direct hazards to the high cost city nodes. 

\begin{figure}
\centering
    \begin{subfigure}[b]{0.48\linewidth}   
\def\svgwidth{0.99\textwidth}
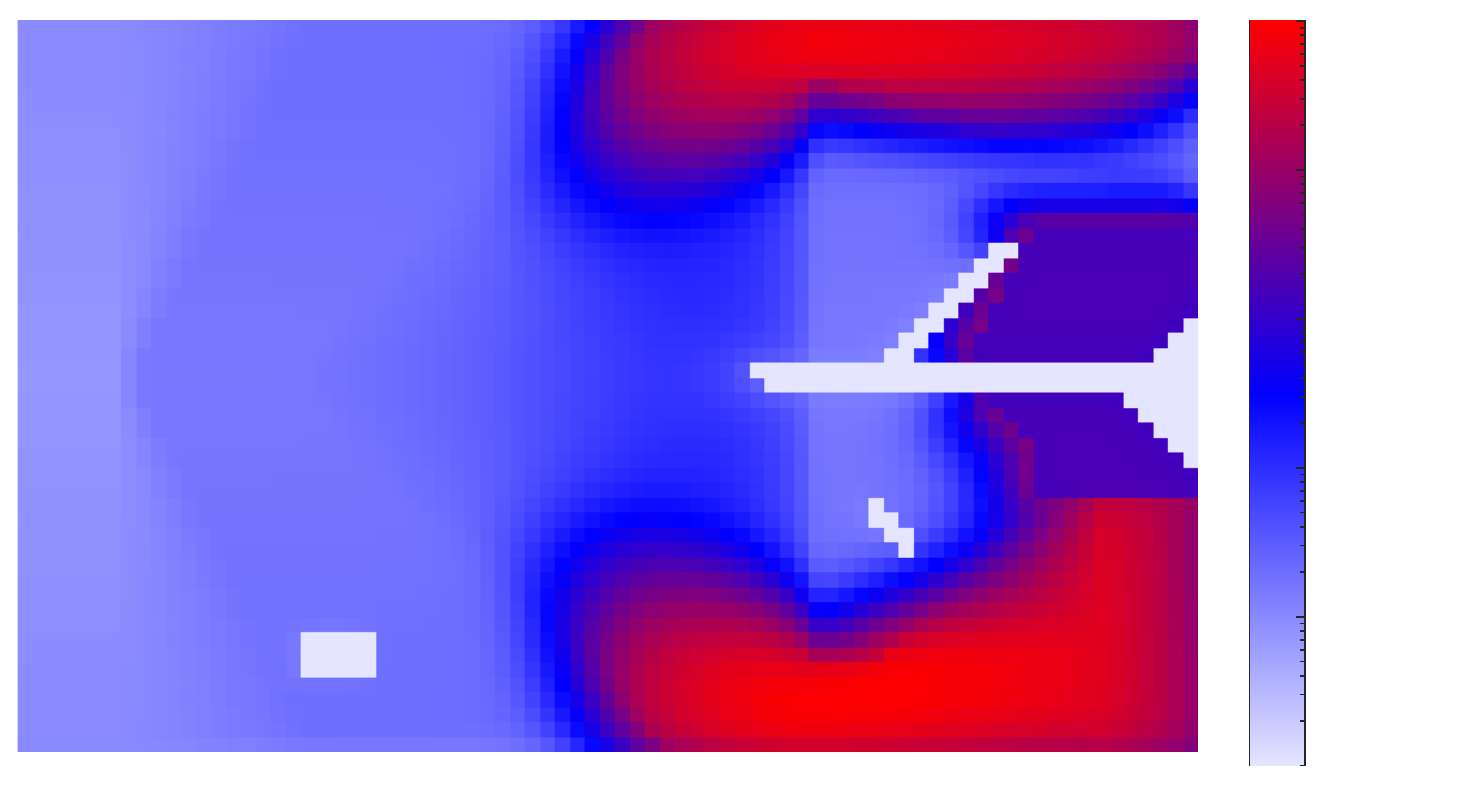          
            \caption{West}
            \label{fig:W22}
    \end{subfigure}%
    \begin{subfigure}[b]{0.48\linewidth}
            \centering
\def\svgwidth{0.99\textwidth}
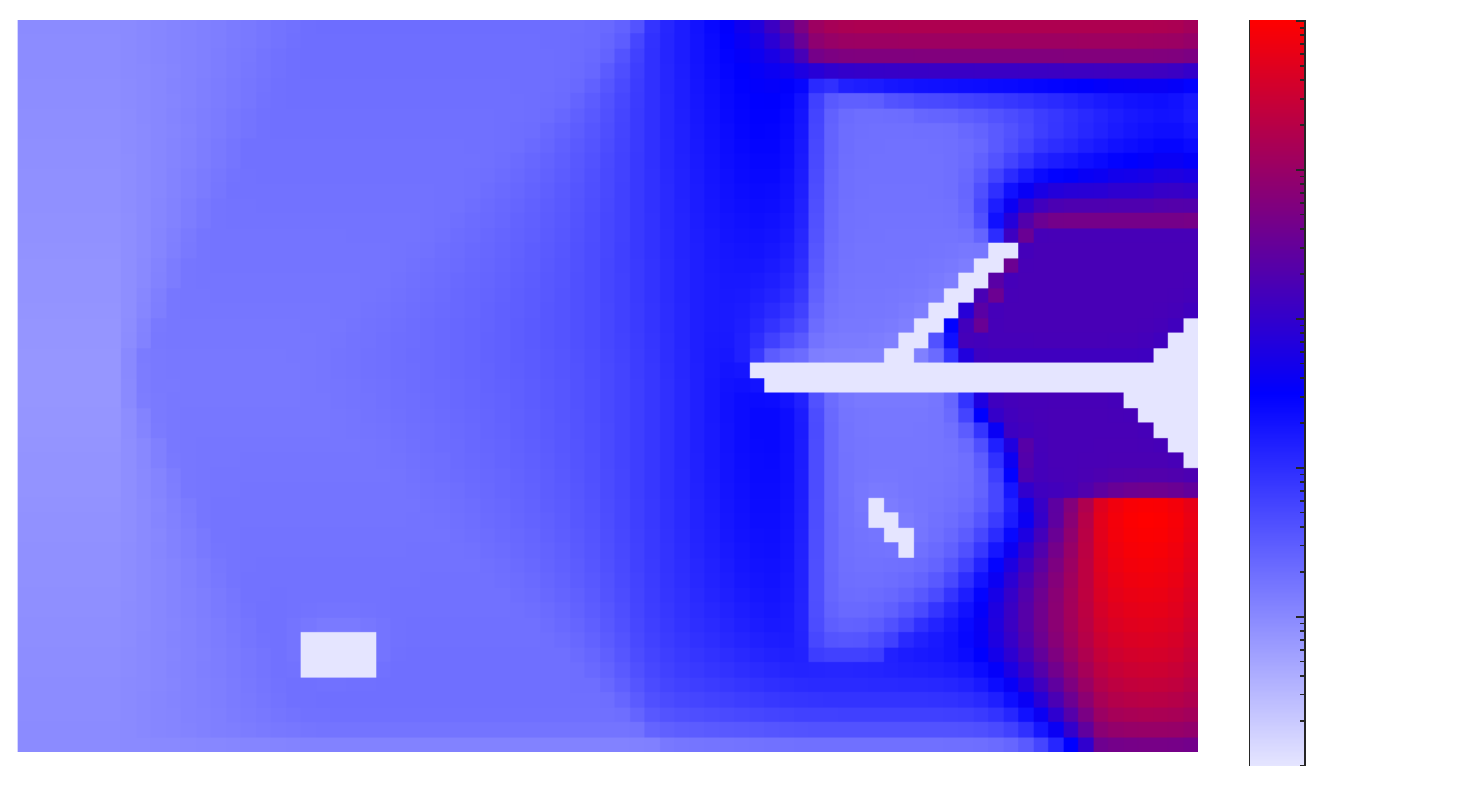    
            \caption{North}
            \label{fig:N22}
    \end{subfigure}
    \begin{subfigure}[b]{0.48\linewidth}            
\def\svgwidth{0.99\textwidth}
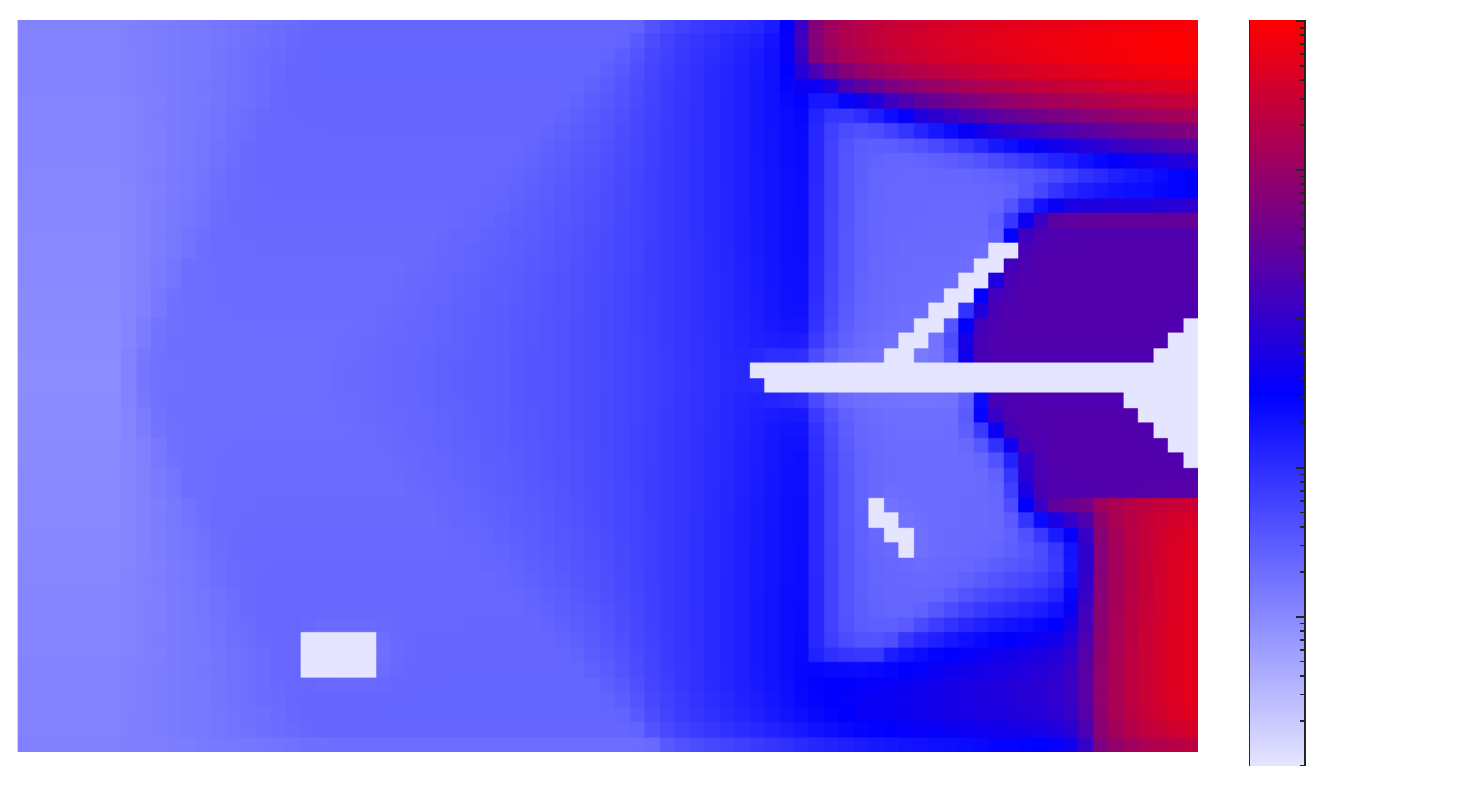    
            \caption{East}
            \label{fig:E22}
    \end{subfigure}%
    \begin{subfigure}[b]{0.48\linewidth}
            \centering
 \def\svgwidth{0.99\textwidth}
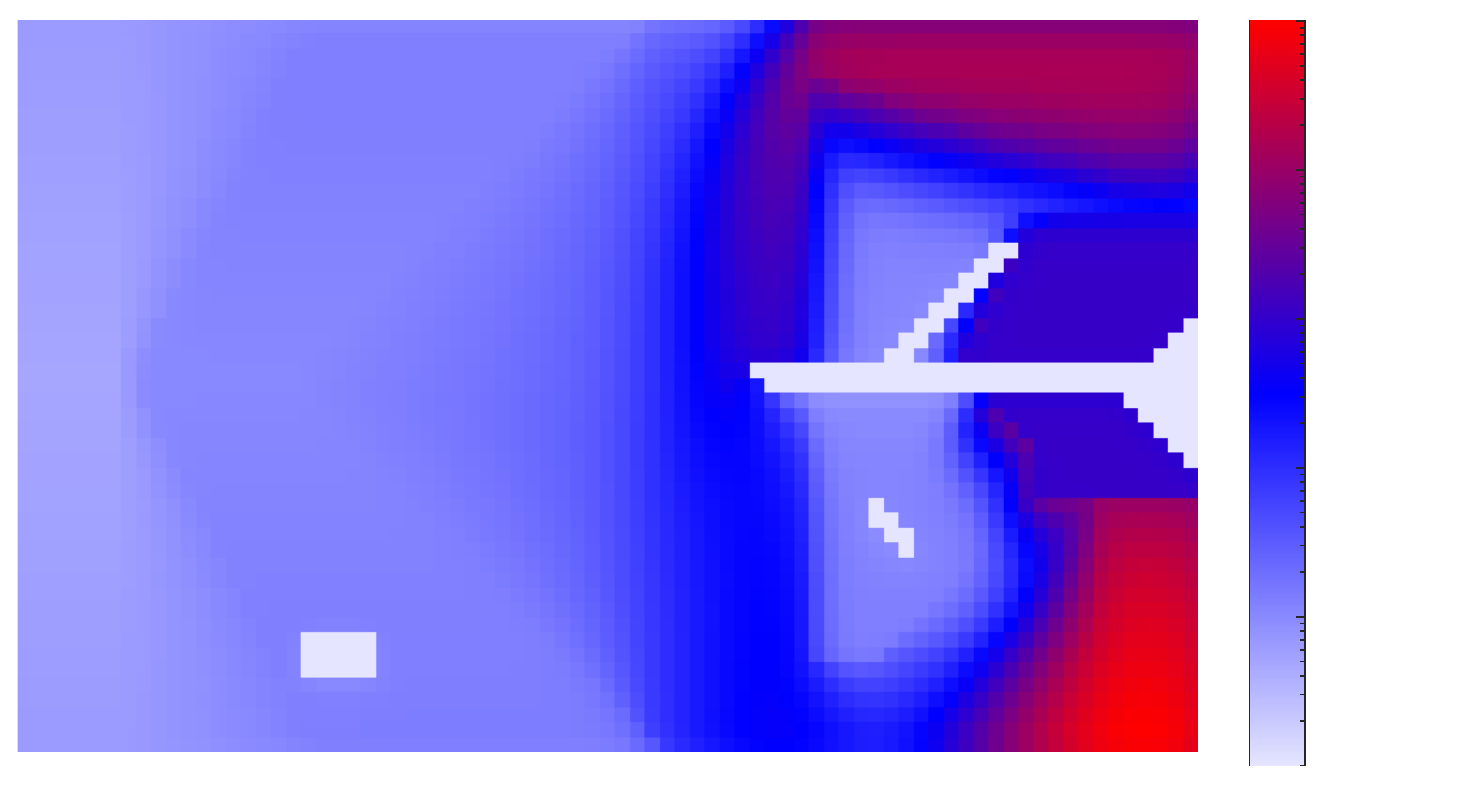    
            \caption{South}
            \label{fig:S22}
    \end{subfigure}
    \caption{Logarithmic surveillance risk map for different wind directions for the landscape in Fig. \ref{fig:NL4000}. Wind direction is indicated as the direction the wind is coming from.}
\label{fig:4000Wind2}
\end{figure}

It should be noted that the generated risk maps are for a revisit rate of $1$, i.e. $\frac{1}{\tau}=1$. If we now want to minimize the maximum risk by 50\% by use of a revisit schedule, i.e. $R_\text{max}=0.5R_0$ where $R_0$ is the risk in Fig. \ref{fig:W22}, we can obtain the smallest revisit rate by using Theorem \ref{th:SP}, assuming $\epsilon_R$ is negligible. The results are displayed in Fig. \ref{fig:SPR}. 

\begin{figure}
\centering
\def\svgwidth{0.85\linewidth}
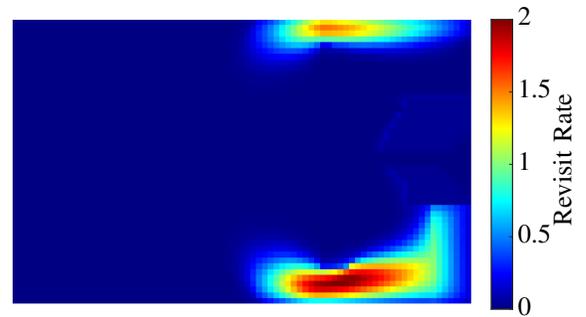   
      \caption{Minimum revisit rate $\frac{1}{\tau_i}$ needed to reduce the maximum risk of Fig. \ref{fig:W22} by 50\%.}
 \label{fig:SPR}
\end{figure}

We now consider minimizing the risk of an undetected outbreak by allocating resources and by use of a revisit schedule as proposed in Problem \ref{P2}. Let us use the same example as the risk map demonstrated in Fig. \ref{fig:W22},  i.e. westerly wind of $V=8$ m/s and $r=4$. Our objective is to minimize the risk \eqref{eq:Rlamb} and we can allocate resources on the spreading rate $\beta$, outbreak rate $\lambda$ and revisit interval $\tau$ \eqref{eq:RM}. Now taking $\overline{\sigma}_i=\text{log}(8)$, $\Gamma_\beta=2000,\Gamma_\lambda=500$ and $\Gamma_\tau=1500$, we obtain the allocation in Fig. \ref{fig:4000BR}. It can be seen that resources are mainly allocated to high spreading areas that also have higher cost. 

\begin{figure}
\centering
     \begin{subfigure}[b]{0.85\linewidth}
    \def\svgwidth{1\textwidth}
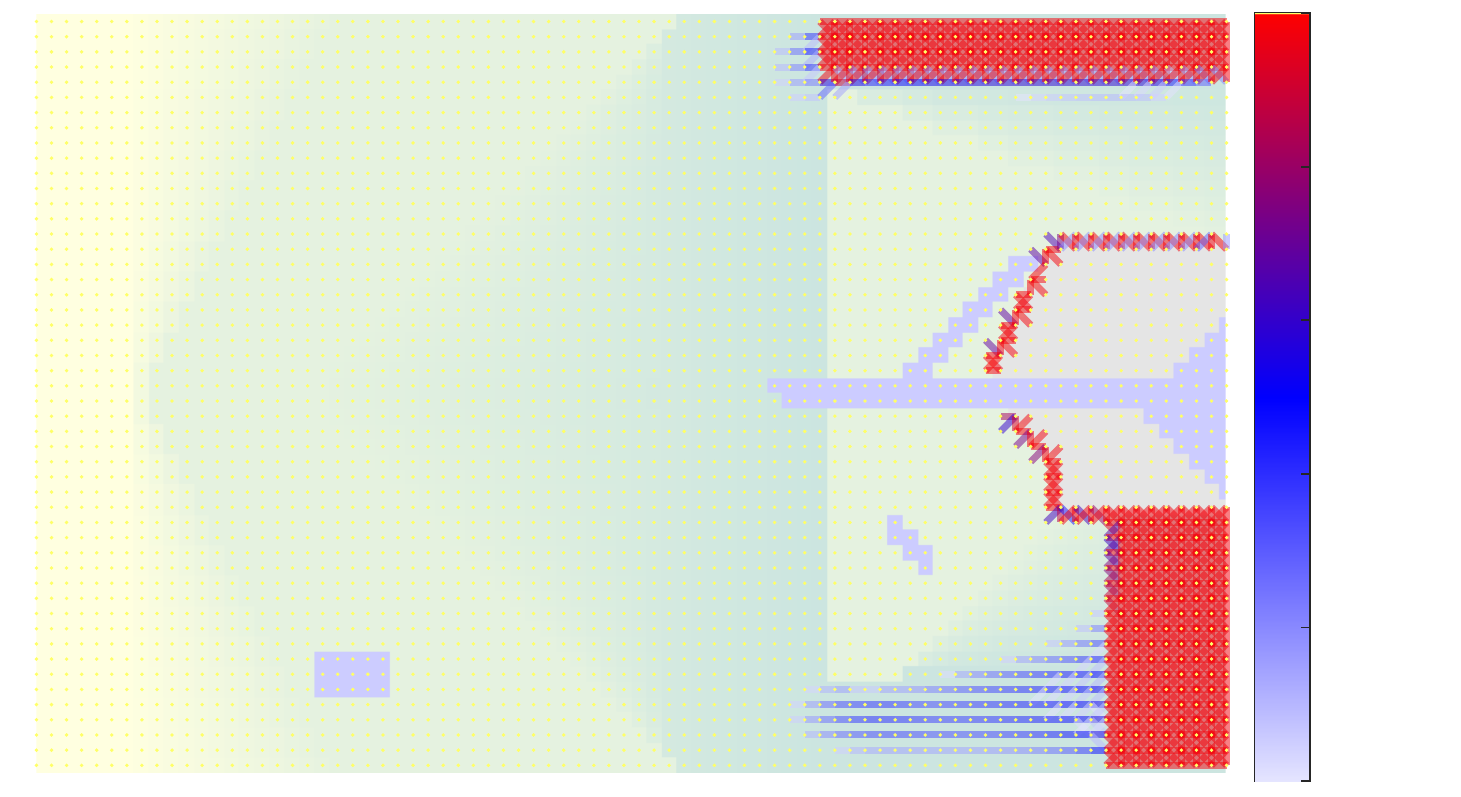  
      \caption{Resource allocation on spreading rate $\beta$}
      \label{fig:4000BRBeta}
  \end{subfigure}
\begin{subfigure}[b]{0.85\linewidth}
     \def\svgwidth{1\textwidth}
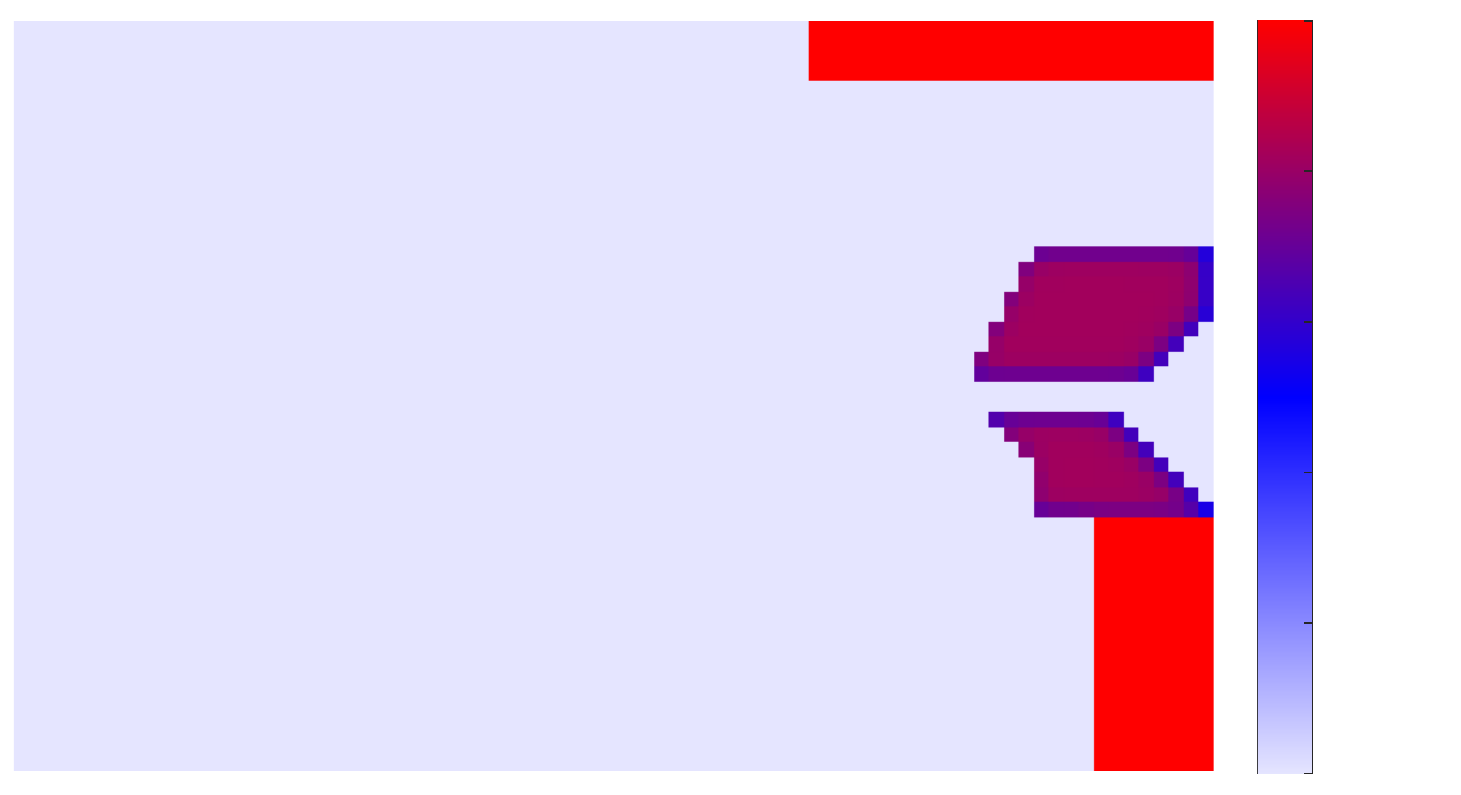
      \caption{Resource allocation on outbreak rate $\lambda$}
      \label{fig:4000BRTau}
  \end{subfigure}
  ~ 
     \begin{subfigure}[b]{0.85\linewidth}
     \def\svgwidth{1\textwidth}
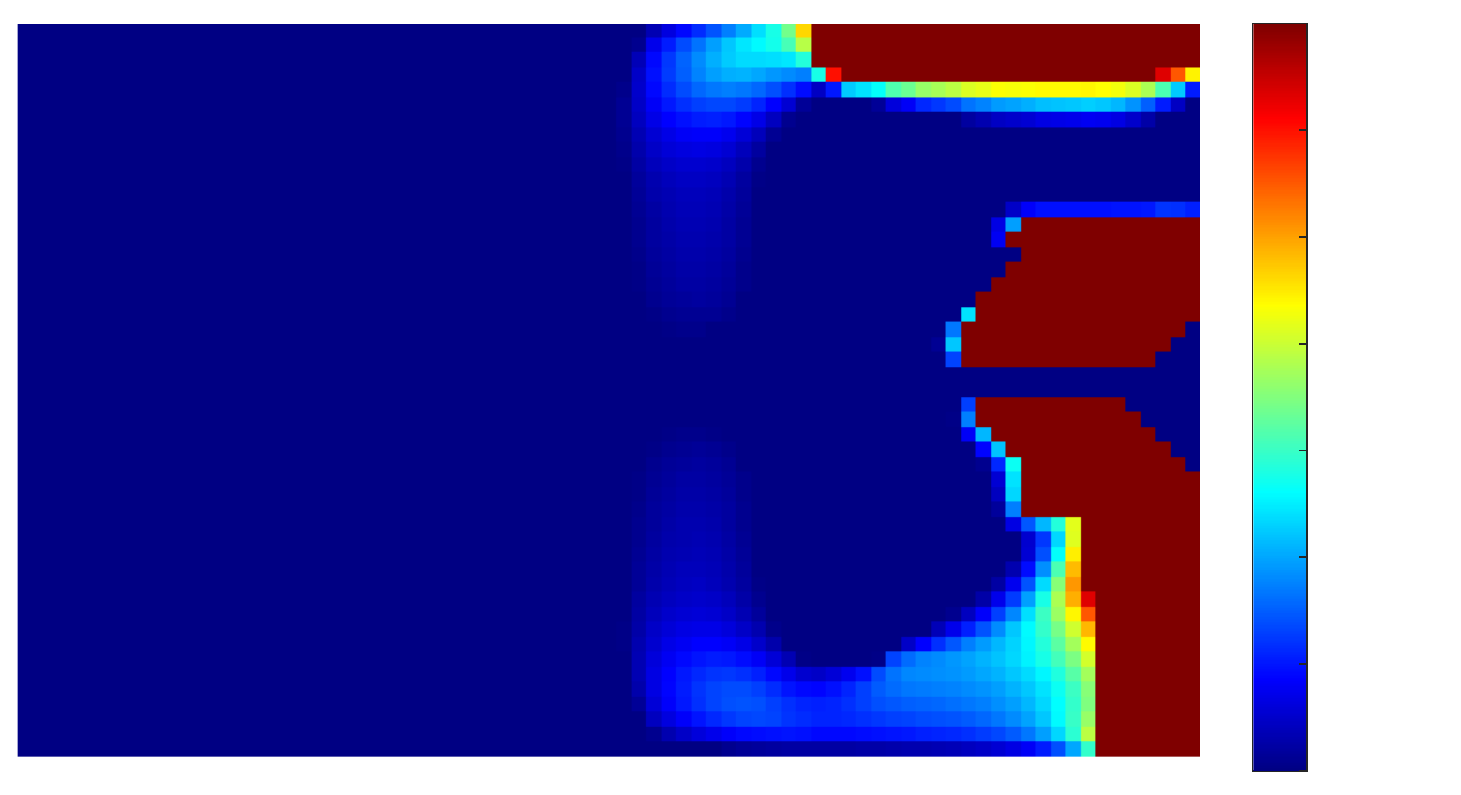
      \caption{Proposed revisit rate $\frac{1}{\tau}$}
      \label{fig:4000BRRevisit}
  \end{subfigure}
      \caption{Objective is to minimize the risk \eqref{eq:Rlamb} by use of resource allocation on $\beta$, $\lambda$ and $\tau$, budget $\Gamma_\beta=2000,\Gamma_\lambda=500$ and $\Gamma_\tau=1500$.}
 \label{fig:4000BR}
\end{figure}

If we want to increase sparsity further, we can use the reweighted $\ell_{1}$ minimization as explained in Section \ref{subsec:L0}. Increasing the upperbound for $\sigma_i$ to log$(16)$, the resulting allocation is visualized in Fig. \ref{fig:4000BRL1} for spreading rate $\beta$. The amount of edges with resources allocated is reduced from $1273$ to $289$ while achieving the same risk bound. 

\begin{figure}[!ht]
\centering
\def\svgwidth{0.85\linewidth}
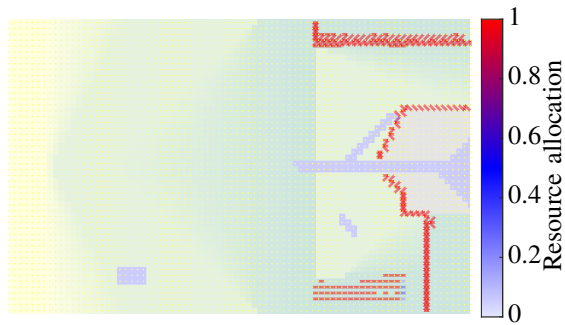
\caption{Resource allocation on spreading rate $\beta$ after using reweighted $\ell_{1}$ minimization on Fig. \ref{fig:4000BR}.}
\label{fig:4000BRL1}
\end{figure}
%

To demonstrate how the proposed method can connect spreading models to path planning, we give an example of using the revisit rate for persistent monitoring. For the persistent monitoring problem different algorithms and solution exists \cite{Pasqualetti2012, smith2012persistent,Alamdari2014a}, but we will use the minimum maximum latency walk from \cite{Alamdari2014a}. This method ranks nodes into different subclasses based on their value, where the higher the value and subclass the more often the node will be visited. Then different walks are created for which a travelling salemesman path is found. 

We assume that the UAV is able to monitor a 4x4 square of nodes at the same time and use the revisit map given in Fig. \ref{fig:4000BRRevisit} as input, where we take the maximum revisit rate of the 4x4 square of nodes rounded to the nearest integer. The minimum maximum latency walk for persistent monitoring as visualized in Fig. \ref{fig:MinMaxWalk} is obtained. 

\begin{figure}
\centering
\includegraphics[width=0.85\linewidth]{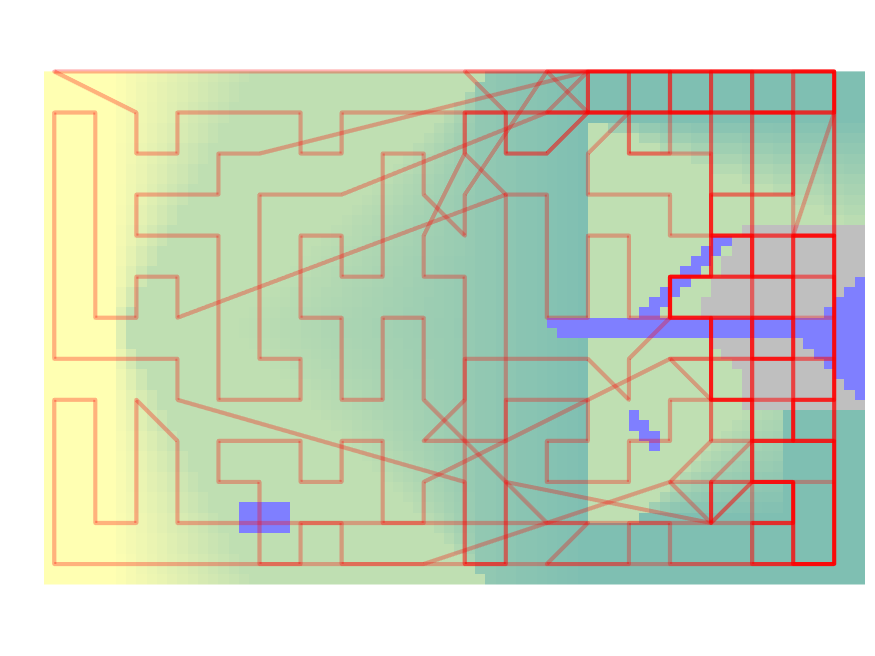}
\caption{Persistent monitoring path consisting of 8 subwalks based on the Revisit Map in Fig. \ref{fig:4000BRRevisit}. Darker lines indicate multiple subwalks.}
\label{fig:MinMaxWalk}
\end{figure}

To validate our proposed method which is based on various bounds and approximations, we simulate the original stochastic model \eqref{eq:Stoc} a 1000 times and compare total cost and warning time. Warning time is taken as time between fire detection and reaching a high cost node (i.e. $c_i\geq0.25)$. Fire detection is simulated by using the path as shown in Fig. \ref{fig:MinMaxWalk}, whereas for the comparison model, the path is taken as a standard lawnmower pattern. We assume that the UAV flies 24 times faster than the fire spread, based on average fire spread rate \cite{Alexandridis2008a,cruz201910} and taking an average UAV speed of $20$m/s \cite{pan2018dynamic}. We compare our proposed resource allocation, as shown in Fig. \ref{fig:4000BR}, with minimizing the spectral radius \eqref{eq:spec}, taking the same resource allocation budget, see Fig. \ref{fig:Scatter}. It can be seen that the proposed method overall has a higher reduction in cost and a longer warning time for higher cost runs. 

\begin{figure}
\centering
 \def\svgwidth{0.85\linewidth}
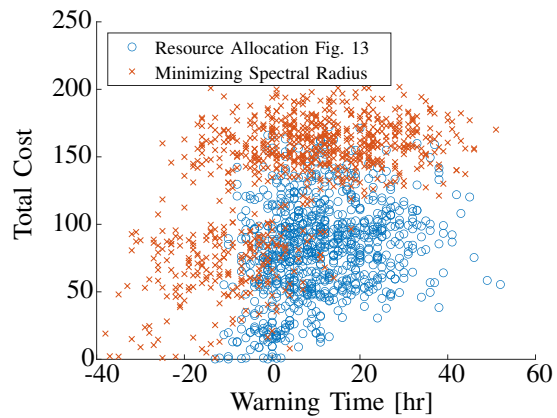
\caption{Comparing total cost and warning time of the stochastic model for our proposed resource allocation, Fig. \ref{fig:4000BR}, versus minimizing the spectral radius \eqref{eq:spec}.}
\label{fig:Scatter}
\end{figure}

\subsection{Computation time}
In this subsection we examine scalability of the proposed method. The bound on the risk and the minimum revisit rate can be calculated directly using \eqref{eq:bound} and the intervention optimization problems were formulated using Julia \cite{bezanson2017julia} and solved with MOSEK v9.2 on an Intel i7, 2.6GHz, 8GB RAM. For the presented examples up to 4000 nodes this is in the order of seconds and as demonstrated in Fig. \ref{fig:Run} the run time grows linearly in the number of nodes, which suggests that the method is scalable to very large networks. 

\begin{figure}
\centering
 \def\svgwidth{1\linewidth}
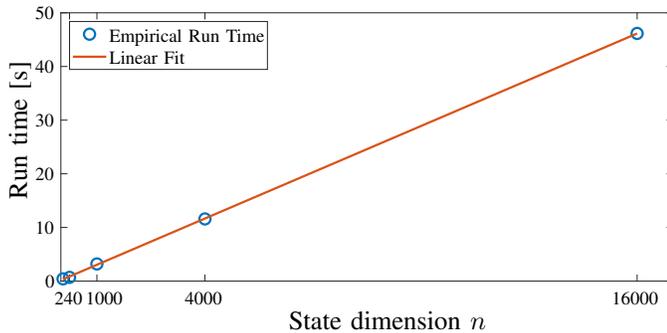
\caption{Run time for Fig. \ref{fig:4000BR} in different dimensions in seconds.}
\label{fig:Run}
\end{figure}
%

\section{Conclusion}
In this paper we presented a flexible optimization framework to bound and minimize the risk of spreading processes by use of both surveillance schedules and sparse control. A risk model was presented that takes into account node dependent costs and risk was defined as the product of the outbreak rate, revisit interval and impact of the outbreak. With multiple examples of different spreading processes we conveyed the use and scalability of the presented framework. 

Future work will include more realistic propagation models and sensing assumptions, multi-UAV planning scenarios and time dependent intervention. 

%
%
\bibliographystyle{IEEEtran}
\bibliography{IEEEabrv,ACFR2021}

\end{document}